%% file: main.tex
\documentclass{article}

\usepackage[letterpaper, margin=1in]{geometry} 
\input{macros}

\title{A Framework for Ruling Out Quantum Speedups}

\date{}

\author{%
\begin{tabular}{c@{\hspace{2.2em}}c}
Thomas Huffstutler$^{\ast}$ & Upendra Kapshikar$^{\dagger}$ \\
David Miloschewsky$^{\ast}$ & Supartha Podder$^{\ast}$
\end{tabular}\\[0.9em]
{\small $^{\ast}$Stony Brook University, USA}\\
{\small $^{\dagger}$University of Ottawa, Canada}\\[0.4em]
{\small $^{\ast}$\texttt{\{thuffstutler, dmiloschewsk, supartha\}@cs.stonybrook.edu}}\\
{\small $^{\dagger}$\texttt{ukapshik@uottawa.ca}}
}

\begin{document}

\maketitle

\begin{abstract}
We study when partial Boolean functions can (and cannot) exhibit superpolynomial quantum query speedups, and develop a general framework for ruling out such speedups via two complementary lenses: \emph{promise-aware} complexity measures and \emph{function completions}. 

First, we introduce promise versions of standard combinatorial measures (including block sensitivity and related variants) and prove that if the relevant promise and completion measures ``collapse'', then deterministic and quantum query complexities are necessarily polynomially related, i.e.,\ $\D(f) = \poly(\Q(f))$. We then analyze structured families of promises, including symmetric partial functions and promises supported on Hamming slices, obtaining sharp (up to polynomial factors) characterizations in terms of a single gap parameter for the symmetric case and refined slice-dependent bounds for $k$-slice domains. 

Next, we formalize \emph{completion complexity} as the minimum of a measure over total completions of a partial function, and show that completability of a measure captures the possibility of superpolynomial quantum speedups. Finally, we apply this viewpoint to derive broad non-speedup criteria for some classes of functions admitting well-behaved completions, such as functions defined on a single slice or a union of slices, functions with low maximum influence on both the standard and $p$-biased hypercubes and functions with efficiently identifiable domains, and then show some hardness results for general completion techniques.
\end{abstract}

\newpage
\tableofcontents

\newpage
\input{intro}

\section*{AI Disclosure}
The authors used ChatGPT 6 Astra to improve the slice bounds in \Cref{sec:slices} and to simply the proof of this result. The manuscript was written entirely by the authors whom take full responsibility for the contents of the paper.

\input{prelims}

\input{perp}

\input{other_results}

\input{extension}

\newpage
\bibliographystyle{alpha}  
\bibliography{main}  

\appendix
\input{appendix}

\end{document}

%% file: macros.tex
\usepackage{amsmath,amssymb,amsthm}
\usepackage{mathtools}

\usepackage{graphicx}
\usepackage{tikz}
\usepackage{tikz-3dplot}
\usetikzlibrary{decorations.pathreplacing,decorations.pathmorphing,arrows.meta}
\usepackage{pgfplots}
\pgfplotsset{compat=1.18}
\usepackage{tikz-cd}
\usepackage{lmodern}

\usepackage{colortbl}
\definecolor{headerblue}{RGB}{220,230,245}

\usepackage{comment}
\usepackage{bbm}
\usepackage{xparse}      
\usepackage{changepage}
\usepackage{algorithm}
\usepackage{algpseudocode}
\usepackage{dirtytalk}
\usepackage[framemethod=TikZ]{mdframed}
\usepackage{extarrows}
\usepackage{arydshln}
\usepackage[most]{tcolorbox}
\usepackage{setspace}
\usepackage[table,dvipsnames]{xcolor}
\usepackage{aliascnt}
\usepackage{xspace}

\usepackage[pagebackref=true]{hyperref}
\usepackage{cleveref}

\renewcommand*{\backref}[1]{}
\renewcommand*{\backrefalt}[4]{%
  \ifcase #1 %
  \or
    \space (cit.\ on p.~#2)%
  \else
    \space (cit.\ on pp.~#2)%
  \fi
}

\pgfplotsset{compat=1.18}

\definecolor{headerblue}{RGB}{220,230,245}  

\newtheorem{theorem}{Theorem}

\newaliascnt{conjecture}{theorem}
\newtheorem{conjecture}[conjecture]{Conjecture}
\aliascntresetthe{conjecture}

\newaliascnt{corollary}{theorem}
\newtheorem{corollary}[corollary]{Corollary}
\aliascntresetthe{corollary}

\newaliascnt{definition}{theorem}
\newtheorem{definition}[definition]{Definition}
\aliascntresetthe{definition}

\newaliascnt{question}{theorem}

\aliascntresetthe{question}

\newaliascnt{example}{theorem}
\newtheorem{example}[example]{Example}
\aliascntresetthe{example}

\newaliascnt{lemma}{theorem}
\newtheorem{lemma}[lemma]{Lemma}
\aliascntresetthe{lemma}

\newaliascnt{fact}{theorem}
\newtheorem{fact}[fact]{Fact}
\aliascntresetthe{fact}

\newaliascnt{claim}{theorem}
\newtheorem{claim}[claim]{Claim}
\aliascntresetthe{claim}

\newaliascnt{observation}{theorem}
\newtheorem{observation}[observation]{Observation}
\aliascntresetthe{observation}

\newaliascnt{proposition}{theorem}
\newtheorem{proposition}[proposition]{Proposition}
\aliascntresetthe{proposition}

\newaliascnt{result}{theorem}
\newtheorem{result}[result]{Result}
\aliascntresetthe{result}

\newaliascnt{remark}{theorem}
\newtheorem{remark}[remark]{Remark}
\aliascntresetthe{remark}

\newaliascnt{problem}{theorem}
\newtheorem{problem}[problem]{Problem}
\aliascntresetthe{problem}

\newcommand{\polyeq}{=_{\mathrm{poly}}}
\newcommand{\polyleq}{\leq_{\mathrm{poly}}}
\newcommand{\dom}{\ensuremath{\mathrm{Dom}}}

\newcommand{\cmplcom}[1]{\overline{#1}}

\DeclareMathOperator{\E}{\mathbb{E}}

\crefname{question}{question}{questions}
\Crefname{question}{Question}{Questions}
\crefname{conjecture}{conjecture}{conjectures}
\Crefname{conjecture}{Conjecture}{Conjectures}
\Crefname{fact}{Fact}{Fact}
\Crefname{theorem}{Theorem}{Theorems}

\usetikzlibrary{decorations.pathmorphing}
\usetikzlibrary{arrows.meta}

\newcommand{\wt}{\ensuremath{\mathrm{wt}}}
\newcommand{\Q}{\ensuremath{\mathrm{Q}}}

\newcommand{\poly}{\ensuremath{\mathrm{poly}}}

\newcommand{\norm}[1]{\ensuremath{\left\lVert #1 \right\rVert}}

\newcommand{\R}{\mathrm{R}}
\newcommand{\bs}{\mathrm{bs}}
\newcommand{\fbs}{\mathrm{fbs}}

\newcommand{\D}{\mathrm{D}}
\newcommand{\Inf}{\mathrm{Inf}}

\newcommand{\T}{\mathrm{T}}
\newcommand{\QO}{\mathrm{O}}
\newcommand{\rc}{\mathrm{r}_{\mathrm{C}}}
\newcommand{\spar}{\mathrm{spar}}
\newcommand{\approxspar}{\widetilde{\mathrm{spar}}}
\newcommand{\Ker}{\mathrm{Ker}}

\newmdenv[
  topline=false,
  bottomline=false,
  rightline=false,
  linewidth=3pt,
  linecolor=black!20,
  innerleftmargin=10pt,
  innerrightmargin=0pt,
  innertopmargin=2pt,
  innerbottommargin=2pt,
  skipabove=6pt,
  skipbelow=0pt,
]{algbox}

\newtcolorbox[blend into=figures]{myfigure}[2][]{float=htb,
title={#2},#1}

\newcommand{\approxdegree}{\ensuremath{\widetilde{\mathrm{deg}}}}
\newcommand{\domain}[0]{\ensuremath{\textup{Dom}}}

\RenewDocumentCommand{\domain}{g}{%
  \ensuremath{%
    \textup{Dom}%
    \IfNoValueF{#1}{\!\left(#1\right)}%
  }%
}
\newcommand{\bxcertificate}[2]{\ensuremath{\mathrm{C}^{#1}(#2,x)}}
\newcommand{\bcertificate}[2]{\ensuremath{\mathrm{C}^{#1}(#2)}}
\newcommand{\certificate}[1]{\ensuremath{\mathrm{C}(#1)}}
\newcommand{\bxperpcertificate}[2]{\ensuremath{\mathrm{C}^{\{#1, *\}}(#2,x)}}
\newcommand{\bperpcertificate}[2]{\ensuremath{\mathrm{C}^{\{#1, *\}}(#2)}}
\newcommand{\perpcertificate}[1]{\ensuremath{\mathrm{C}^{\prime}(#1)}}
\newcommand{\xsensitivity}[2]{\ensuremath{\mathrm{s}(#2, #1)}}
\newcommand{\sensitivity}[1]{\ensuremath{\mathrm{s}(#1)}}
\newcommand{\xblocksensitivity}[2]{\ensuremath{\mathrm{bs}(#2, #1)}}
\newcommand{\blocksensitivity}[1]{\ensuremath{\mathrm{bs}(#1)}}
\newcommand{\xperpsensitivity}[2]{\ensuremath{\mathrm{s}^{\prime}(#2, #1)}}
\newcommand{\perpsensitivity}[1]{\ensuremath{\mathrm{s}^{\prime}(#1)}}
\newcommand{\xperpblocksensitivity}[2]{\ensuremath{\mathrm{bs}^{\prime}(#2, #1)}}
\newcommand{\perpblocksensitivity}[1]{\ensuremath{\mathrm{bs}^{\prime}(#1)}}
\newcommand{\criticalblocksensitivity}[1]{\ensuremath{\mathrm{cbs}(#1)}}
\newcommand{\balancedblocksensitivity}{\ensuremath{\mathrm{bs}}}

\newcommand{\deterministicquery}[1]{\ensuremath{\mathrm{D}(#1)}}
\newcommand{\FC}{\ensuremath{\mathrm{FC}}}   
\newcommand{\RC}{\ensuremath{\mathrm{RC}}}   
\newcommand{\polyequal}[2]{ \ensuremath{#1 =_{\mathrm{poly} }#2}}
\DeclarePairedDelimiter\abs{\lvert}{\rvert}
\newcommand{\criticalcertificate}[1]{\ensuremath{\mathrm{C}_{\mathrm{cri}}(#1)}}
\newcommand{\gap}{\mathrm{GAP}}

\newcommand{\slice}[1]{S_{#1}}
\newcommand{\unionSlice}[1]{X_{#1}}

\renewcommand{\dom}{\domain}
\newcommand{\Dom}{\dom}

\renewcommand{\bs}{\mathrm{bs}}

\usepackage[normalem]{ulem}

\newcommand{\pf}{\textup{PF}\xspace}
\newcommand{\lma}{\textup{LMA}\xspace}
\newcommand{\sat}{3-SAT\xspace}
\newcommand{\NP}{{\ensuremath{\mathsf{NP}}\xspace}}

%% file: intro.tex
\section{Introduction}

One of the central questions in quantum complexity theory is identifying problems which exhibit superpolynomial quantum speedups.
While polynomial quantum speedups are relatively common and can be constructed for many problems, finding exponential quantum speedups is significantly more challenging. 
When it comes to showing exponential quantum speedups and/or designing quantum algorithms, the query complexity model is often the preferred choice, as it is simple enough to facilitate proving lower bounds, yet expressive enough to capture the varying levels of problem difficulty.
In this model, query access is given to an unknown $n$-bit input $x$ with the goal of computing a predefined function $f$ on $x$. The queries are given in a bit-wise manner, i.e., given an index $i$ as a query, we get $x_i$ as the response. The goal is to minimize the number of queries necessary to compute $f(x)$ in the worst case. In the randomized setting, algorithms are allowed to output incorrect answers with probability $\tfrac{1}{3}$. When we consider quantum algorithms in this model, we allow the algorithms to make queries in superposition, thus raising the question of how much advantage quantum algorithms provide over classical algorithms in this setting.

Over the years, only a handful of structured problems---such as period-finding in Shor’s algorithm or Simon's problem~\cite{shor1999polynomial, simon1997power}---have been shown to admit exponential quantum speedups. In contrast, unstructured problems generally allow only polynomial speedups.
Despite being an important problem, only a few works have attempted to formalize quantum speedups and characterize what constitutes a superpolynomial quantum speedup. 
It was shown by~\cite{beals2001quantum} that when Boolean functions are defined on all inputs (i.e., $f: \{0,1\}^n \rightarrow \{0,1\}$), quantum and deterministic query complexities are polynomially (power 6) related. The relation was later improved to power 4 by~\cite{aaronson2021degree}, using Huang's sensitivity theorem~\cite{Hua19}. 
\cite{nisan1989crew} showed that deterministic and randomized complexities are polynomially related for total functions and
in fact, for total functions, all well-studied query complexity measures are polynomially related\footnote{A table of polynomial relationships between well-studied measures for total functions can be found in~\cite{aaronson2021degree}}~\cite{buhrman2002complexity}. 
Thus, to achieve superpolynomial quantum speedups, we must consider partial functions (i.e., Boolean functions defined only on a subset of inputs, $f: X \subseteq \{0,1\}^n \rightarrow \{0,1\}$). 
In fact, we know\footnote{This is because $\D(f) = O\left(\Q(f)^2 \cdot \log^2 |\domain(f)| \right)$~\cite{aaronson2015sculptingquantumspeedups}.} that the domain of any function achieving superpolynomial quantum speedup needs to be large~\cite{aaronson2015sculptingquantumspeedups}. 

However, not all partial functions (even on large domain size) exhibit quantum speedups. Thus, an important question is determining which partial functions exhibit superpolynomial quantum speedups, which do not, and more importantly, why. In pursuit of this, \cite{Need_For_Structure} showed that any partial symmetric function does not exhibit such speedup\footnote{They showed it for randomized vs.\ quantum.}. \cite{chailloux2019note} later strengthened this result to a more general notion of symmetry and showed that certain types of permutation invariant functions admit no superpolynomial speedup. 
However, symmetric functions appear highly irregular as the function remains invariant under all possible permutations of the indices. \cite{ben2024symmetries}~raised the question of how symmetric must a function be before it cannot exhibit a large quantum speedup and characterized functions invariant under different group actions, studying which ones allow superpolynomial speedups and which do not.

On the other hand, symmetric partial functions, even under various group actions, are only a tiny fraction of all possible partial functions. To shed light on the quantum speedups of non-symmetric functions,~\cite{bendavid2014structurepromisesquantumspeedups} showed that functions defined on the slice, i.e., $f: \binom{[n]}{k} \rightarrow \{0,1\}$, for $k\in[n]$, do not admit superpolynomial quantum speedups. This can be viewed as a generalization of~\cite{beals2001quantum}. Arguably in this case, the structure is on the domain, where every input $x$ is \emph{promised} to have Hamming weight $k$, but the function is free to take any possible value on these inputs. 

In another direction, \cite{aaronson2015sculptingquantumspeedups} studied the quantum advantage question from a different perspective. 
Given a problem which is intractable for both quantum and classical algorithms, can we find a sub-problem for which quantum algorithms provide an exponential advantage? 
Towards this end, they introduced the idea of \emph{sculpting} functions, 
and asked if it is possible to turn a total function into a partial function by removing some of the inputs, so that the newly formed partial function has a superpolynomial quantum speedup. They characterized the total Boolean functions for which an exponential quantum speedup can always be found in this manner and,
as an implication, showed that when the domain size of a partial function is small, i.e., it contains only $\poly(n)$ inputs, it cannot admit superpolynomial quantum speedups.

Other related works, such as \cite{yamakawa2024verifiable, ben2024separations, bendavid2014structurepromisesquantumspeedups, balaji2015graph, kulkarni2016quantum, podder2025fine}, explore quantum speedups from other perspectives, such as in the context of relational problems or under fixed query budgets etc.
Before the work of \cite{yamakawa2024verifiable}, it was widely believed that achieving exponential quantum advantage required some form of global regularity or structure. However, their result suggests that quantum speedup is more subtle and less well-understood than previously thought.
In this work, we revisit Boolean functions and provide alternative approaches to proving the absence of superpolynomial quantum speedups.

\subsection{Our contributions}

For the first result, we look at the promise version of block sensitivity. 
Consider a function $f: \lbrace 0,1 \rbrace ^n \to \lbrace 0,1 \rbrace $. A \emph{block} $B$ is a subset of indices, $B \subset [n]$. 
For a total function, a block $B$ is considered sensitive for input $x$ if $f(x) \neq f(x^B)$, where $x^B$ is obtained by flipping coordinates in $B$. For partial functions, we relax this notion: a block is \emph{also} considered sensitive if flipping it causes the input to leave the promise (the function output is undefined)~\cite{chakraborty2022certificate}. We call this measure \emph{promised block sensitivity} and label it $\perpblocksensitivity{f}$. A formal definition of $\perpblocksensitivity{f}$ and other measures discussed may be found in~\Cref{sec:partial_function_measures}. For two measures $A$ and $B$, we use $A \polyeq B$ to denote that they are polynomially related. Using the promise version of block sensitivity we show that large quantum query advantages cannot arise when other complexity measures collapse.

\begin{result}[\textbf{Requirements for superpolynomial quantum speedup};~\Cref{thm:measures_no_speedup}]\label{res:promise_measure_speedup}
    For any partial Boolean function $f$, if either (i) critical block sensitivity \criticalblocksensitivity{f} is polynomially related to \emph{promised block sensitivity} \perpblocksensitivity{f}, or (ii) critical block sensitivity is polynomially related to critical certificate complexity \criticalcertificate{f}, then quantum and deterministic query complexities are polynomially related. Formally,
    \begin{align*}
         \criticalblocksensitivity{f} \polyeq \perpblocksensitivity{f}&\textbf{ or } \criticalblocksensitivity{f}\polyeq \criticalcertificate{f}\\
        &\;\Downarrow \\
        \D(f)&\;\polyeq \Q(f).
    \end{align*}
\end{result}

Note that there are multiple ways one may define measures for partial functions which are equivalent for total functions~\cite{chakraborty2022certificate, anshu2020querytocommunicationliftingadversarybounds}, we chose our definitions as we found them the most natural for studying the role of $\domain{f}$ when comparing $\deterministicquery{f}$ and $\Q(f)$.

Many natural Boolean functions exhibit a high degree of permutation symmetry: their value depends only on the Hamming weight of the input, or they could be defined only on a single orbit under the action of the symmetric group $S_n$.
Formally, we study two related families of partial functions with $S_n$-symmetry:
\begin{itemize}
    \item \emph{fully symmetric} partial functions whose value depends only on $\mathrm{wt}(x)$, and
    \item functions whose domain is a single $k$-slice of the Boolean cube (equivalently, an $S_n$-orbit), where the function may be arbitrary on that slice.
\end{itemize}
Our first set of results gives a clean, tight characterization (up to polynomial factors) of several query complexity measures for nonconstant partial symmetric functions in terms of a single parameter we denote {$\gap$},
\[
\gap(f)\;=\;\min_{x,y\in\domain(f),\,f(x)\neq f(y)} |\mathrm{wt}(x)-\mathrm{wt}(y)|.
\]
Informally, $\gap(f)$ captures the minimal Hamming-weight separation between any two inputs on which the function differs. 

\begin{result}[\textbf{Partial functions and $S_n$}; Lemma~\ref{lem:Q_for_sym}, \ref{lem:d_for_sym}]\label{res:symmetric}
For any nonconstant partial symmetric function~$f$,
\begin{align*}
    \approxdegree(f) &\polyeq \Q(f) \polyeq \R(f) \polyeq \frac{n}{\gap} \hspace{0.5cm} \text{ and} \\
    \D(f) &\polyeq n - \gap + 1.
\end{align*}
Combining these results, we get a complete picture of when quantum speedups over deterministic algorithms are possible for symmetric functions. As a corollary, we also obtain that, for partial symmetric functions, the polynomial method is essentially tight: $\approxdegree$ and $\Q$ are polynomially related.
\end{result}

We then continue the line of work of sculpting~\cite{aaronson2015sculptingquantumspeedups}, but in the reverse direction, i.e., instead of removing inputs from a total function, we will start with a partial function and add inputs to its domain to obtain a total one. This method is known as a \emph{completion} and has been studied in the communication complexity model~\cite{goos2025pseudodeterministic, gavinsky2025unambiguous}. {However, as far as we know, in the query-complexity model, a systematic study of \emph{completions} has been limited in the literature.}

We ask how much different complexity measures can increase under this transformation. 
More specifically, given a partial Boolean function $f$, a total Boolean function $F$ is a
\emph{completion} of $f$ if $F(x)=f(x)$ for all $x \in \domain(f)$, while values on
inputs outside the promise are assigned arbitrarily. Naturally, a partial function may
admit many distinct completions. The \emph{completion complexity} of $f$ with respect to any query complexity measure $M$ is then defined as $ \cmplcom{M}(f) := \min M(F)$, where the minimum is over all possible completions of $f$.

As our third result, we show that completions exactly characterize superpolynomial quantum speedups.

\begin{result}[\textbf{Completions characterize superpolynomial quantum speedups}; \Cref{{lemma:general}}]\label{res:completions}
        For any \emph{partial} function $f$, and for any\footnote{Any measure $M$ which for total functions is polynomially related to $\D$ \cite{beals2001quantum}.} query complexity measure $M,$ deterministic query complexity $\D(f)$ is polynomially related to $M(f)$ if and only if the measure $M$ is \emph{extendable}, that is, 
    \[ \D(f) \polyeq M(f) \textit{ if and only if
 } \cmplcom{M}(f) \polyeq M(f).  \]
        
\end{result}

Specifically, a partial function extends to a total function while preserving a polynomial relationship in any complexity measure, if and only if no superpolynomial speedup over deterministic query complexity is possible for that measure. All details related to this result can be found in \Cref{sec:completion}. The $k$-slice, which is the set of inputs with Hamming weight $k$, and the union over slices with weights in $K\subseteq [n]$ are denoted using $\slice{n,k}$ and $\unionSlice{K}$, respectively. We use \Cref{res:completions} to establish the following theorem regarding the query complexity of partial functions defined on $\slice{n,k}$ and $\unionSlice{K}$.

\begin{result}[\textbf{Completions strengthen polynomial relationships for general slice domains}; \Cref{lemma:slice,thm:single_slice,thm:multislice}]\label{res:slice}
    For any partial functions $f:\slice{n,k} \to \{0,1\}$ and $g:\unionSlice{K}\to \{0,1\}$,
    \begin{align}
        \D(f) &= O\left( \approxdegree(f)^2 \bs(f)^4 \right)\nonumber\\
        \R(g) &= O\left( \bs(g)^2 + \max_{k\in K} \D(g|_{\slice{n,k}})\right).\label{eq:union_slice_result}
    \end{align}
    Consequently, $\D(f)=O(\Q(f)^{10})$ and $\R(g) = O(Q(g)^{10})$. In addition, we also show the sharper result for middle slices with weight close to $\frac{n}{2}$, that for any function on the $k$-slice, 
\begin{align*}
    \D(f)\;\le\;\frac{3n}{\min\{k,n-k\}}\;\certificate{f}\;\balancedblocksensitivity(f),
\end{align*}
which, together with known relationships between $\certificate{f}$, $\balancedblocksensitivity(f)$ and $\Q(f)$, implies $\D(f)=O\!\big(\tfrac{n}{\min\{k,n-k\}} \, \Q(f)^6\big)$.
\end{result}

\Cref{res:slice} is a strengthening of the previous bound on the slice by~\cite{bendavid2014structurepromisesquantumspeedups}, who showed a power $18$ relationship. Note that one cannot replace $\R(g)$ with $\D(g)$ in \Cref{eq:union_slice_result} as the union of slices generalizes symmetric functions which we analyzed in \Cref{res:symmetric}.

We can also use the completions technique to establish quantum non-speedup in the following way: For any (possibly partial) function $f$, $\D(f) \geq \Q(f) \geq \approxdegree(f)/2$. If a partial function $f$ has approximate degree $T$, then by definition for every input $x \in \domain(f)$, the approximating polynomial $p$ outputs a value $p(x)$ with $|p(x)-f(x)| \leq \epsilon$. However, for the inputs $x \notin \domain(f)$, the polynomial can obtain any value in $[0,1]$. Thus, if we can construct a \emph{completed} polynomial $P$ of degree $\poly(T)$ such that $|P(x)-f(x)| \leq 1/3$ for every $x\in \domain(f)$ and $P(x)$ is $1/3$ close to either 0 or 1 for all $x\notin \domain(f)$, then $\D(f) = O({\poly}(\approxdegree(f)))$. See \Cref{fig:polynomial} for an illustration of this concept.

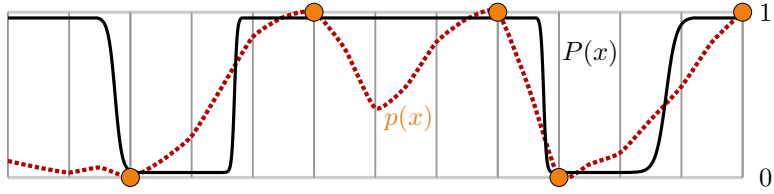
\begin{figure}[h]
    \centering
    \input{poly_approx}
    \caption{The orange points represent the original inputs in the domain, and the red dotted polynomial $p(x)$ is a polynomial that approximates the function $f$ on these points. The solid black polynomial $P(x)$ is a completion that provides an approximation across the entire Boolean hypercube.}
    \label{fig:polynomial}
\end{figure}

Using the completion framework, we show that superpolynomial quantum speedups exist only when the complexity of identifying the domain is superpolynomial with respect to the approximate degree.

\begin{result}[\textbf{No superpolynomial quantum speedups for functions with efficiently verifiable domains};~\Cref{lemma:naive_completion}]
    Let $X \subseteq \{0,1\}^{n}$ be a domain whose membership can be checked by an algorithm in $T$ queries. Then, for \emph{all} functions $f: X \rightarrow \{0,1\}$ with approximate degree $\Omega(\poly(T))$,
    \begin{align*}
        \D(f)\polyeq\approxdegree(f).
    \end{align*}
\end{result}

Additionally, we also demonstrate the usefulness of completions by proving quantum non-speedup results for various function families. We consider various smoothness properties of polynomials, the main property being Lipschitz continuity which has been studied in non-smooth and nonconvex optimization~\cite{Larson_2021}, stochastic convex optimization~\cite{ShalevShwartz2009StochasticCO}, and geometric analysis and optimal transport~\cite{magnani2009contactequationslipschitzextensions}. Furthermore, we consider functions $f$ with low influence as well as functions which are defined on the generalized $p$-biased hypercube. Functions on the $p$-biased hypercube have been studied for their relationships to the Erdős–Rényi random graph model in~\cite{Sousi_2020, nagda2025optimaldistinguishersplantedclique}, as well as with regards to the effects of the $p$-biased hypercube on noise stability~\cite{noisey_p-biased} and linearity testing~\cite{khot2025biasedlinearitytesting1}. These $p$-biased functions have also been studied with respect to their relationships to juntas and low degree functions~\cite{biased_agreement_test, Dinur_2024}, have been used to establish a hypercontractive inequality for general $p$ and prove an analog of the invariance principle of Mossel, O'Donnell and Oleszkiewicz for the $p$-biased hypercube~\cite{mossel2005noisestabilityfunctionslow, keevash2021globalhypercontractivityapplications}.

\begin{result}[\textbf{No superpolynomial quantum speedups for functions with low influence and sparsity};~\Cref{thm:natural_influence} and~\Cref{lem:biased}]
    Let $f:\domain(f)\rightarrow \{-1,1\}$ have an approximating polynomial $p(x)$ such that for all $i\in [n]$, the influence $\Inf_i(p)$ and sparsity $\spar_i(p)$ are low. Then,
    $\D(f)\polyeq \Q(f).$
  
    Furthermore, we generalize the results to the $p$-biased hypercube.
\end{result}

Finally, we propose a general completion technique concerning finding a perturbation vector $\Delta$ for an approximating polynomial $p(x)$ for a partial function $f$, so that adding this perturbation to the Fourier coefficients of $p(x)$ results in a polynomial $p_{\Delta}(x)$ such that $\frac{\min_{x\in\{-1,1\}^n}|p_{\Delta}(x)|}{\max_{x\in\{-1,1\}^n}|p_{\Delta}(x)|}\geq \frac{1}{\poly\left(\approxdegree(f)\right)}$, resulting after amplification in an approximating polynomial for a total function. Furthermore, we show via a reduction from \sat that the general perturbation decision problem is \NP-complete.

\begin{result}
    Given multilinear polynomials $p,q_{1},\dots,q_{n}: \{-1,1\}^{t}\rightarrow \mathbb{R}$ by their full tables of binary-encoded rational values, and rational constants $\epsilon,B>0$ encoded in binary, deciding whether there exists a vector $\Delta\in\mathbb{R}^{n}$ such that for all $x \in \{-1,1\}^t$,
    \begin{align*}
        \biggl| p(x)+\sum_{j=1}^{n}\Delta_{j}q_{j}(x) \biggr| \geq \epsilon \text{  and $||\Delta||_{\infty}\leq B$ } 
    \end{align*}
    is \NP-complete.
\end{result}

\subsection{Our techniques}

\noindent\textbf{Promise measures}

To obtain the requirements for speedup using the two variations of promise measures, we follow the steps of the proof that for total functions $f$, $\polyequal{\deterministicquery{f}}{\Q(f)}$. For a survey of the proofs, see~\cite{BUHRMAN200221}. First, we define the alternative promise versions of  $\blocksensitivity{f}$ and $\certificate{f}$, $\perpblocksensitivity{f}$ and $\perpcertificate{f}$. As $\perpblocksensitivity{f}$ was already described alongside Result~\ref{res:promise_measure_speedup}, let us describe $\perpcertificate{f}$. Standardly, a certificate $c$ for a promise function certifies that each input in the promise which agrees with $c$ has the same output. The alternative $\perpcertificate{f}$ is defined such that no input which agrees with $c$ is outside of $\domain{f}$.

Transferring the relationships from the total to our promise versions of the measures relies on the role of inputs outside of the promise. Specifically, many steps use the assumption that changing the input, for example by flipping multiple indices, keeps it in the promise, which is not necessarily true. For example, consider the proof that the size of a minimal block is bounded by sensitivity. If we attempt to show a promise-version of such a statement, meaning that the size of a minimal sensitive block $B$ such that $x,x^B\in \domain(f)$ is bounded by its sensitivity, we find that removing an index $i$ from $B$ might take it out of the promise, meaning $x^{B\setminus \{i\}} \not\in \domain(f)$ and therefore $i$ might not be a sensitive index. This is underscored by the fact that there are partial functions for which $\sensitivity{f}=0$, such as functions on the slice. On the other hand, using the same argument, we bound the size of a block which certifies that takes an output outside of the promise or changes its output by $\perpsensitivity{f}$ ($\perpblocksensitivity{f}$ restricted to blocks of size $1$). The other proofs follow similarly, meaning that one version of the measure makes the proof impossible, while the other does not.

\noindent \textbf{Efficiently identifying domains is sufficient}

When trying to think of explicit constructions for a completion, the na\"ive idea is to just set all inputs not in the promise of a partial function to have an output of $0$. To characterize the complexity in using this na\"ive completion, we examine the difficulty of determining, for any partial function $f:X\rightarrow \{0,1\}$, whether a given input $x\in X$ or if $x\notin X$. This is equivalent to the condition that there exists a low degree approximation of the indicator function $\mathbbm{1}_{X}$ on $\{0,1\}^n$, denoted by $\widetilde{\mathbbm{1}}_{X}$. We take the degree $d$ $\epsilon-$approximating polynomial $p(x)$ for $f(x)$, and the low degree ($\leq \poly(d)$) $\epsilon'-$approximator $q(x)=\widetilde{\mathbbm{1}}_{X}$ and define the new polynomial $r(x)=q(x)p(x)$. This new polynomial approximates the na\"ive completion $F$ of $f$ everywhere on the Boolean hypercube to error at most $\epsilon+\epsilon'$ by the triangle inequality. Since $\widetilde{\mathbbm{1}}_{X}$ is defined for a total function, this implies that being able to efficiently (classically) distinguish between $X$ and $\{0,1\}^{n}/X$ is sufficient to show that the na\"ive completion is enough to prove that no superpolynomial speedup exists for the given partial function $f$.

\noindent \textbf{Characterizing the natural completion}

A natural completion one can think of is based on the sign of an  approximating polynomial. If we are given a partial function $f:\domain{f}\rightarrow \{-1,1\}$ which has a degree $d=\approxdegree(f)$ approximating polynomial $p(x)$, then $p(x)$ will still output a value on inputs $x\notin \domain{f}$. If we take any $x\notin \domain{f}$ and evaluate $p(x)$, then as long as $p(x)$ is at least $1/\poly(d)$ away from $0$ on the hypercube, by standard boosting techniques we can transform $p(x)$ into a degree $\poly(d)$ approximating polynomial for the completion $F(x) = \textup{sign}(p(x))$ (i.e. $F(x)=1$ if $p(x)>0$ and $F(x)=-1$ if $p(x)<0$). The problem is that we are only guaranteed this holds for $x\in\domain{f}$. For $x\not\in\domain{f}$, we could have $\abs{p(x)}\leq \epsilon$ where $\epsilon \ll \poly(d)^{-1}$.

We would like to characterize the notion of this \emph{natural} completion more precisely and give a condition which captures when we can guarantee that $\abs{p(x)}\geq \frac{1}{\poly(\approxdegree(f))}$ on inputs outside of the domain. To this end, we consider global smoothness properties of polynomials, namely the Lipschitz continuity of an approximating polynomial for $f$. We give a sufficient Lipschitz continuity condition for an approximating polynomial $p(x)$ in order to ensure that $|p(x)|\geq \frac{1}{\poly(\approxdegree(f))}$ so that we can boost $p(x)$ to the canonical $\frac{1}{3}$ error approximating polynomial for $F$ while only adding a polynomial blow-up in the degree of $p(x)$. We then give conditions on the influence and sparsity of partial functions $f$ that ensure this global smoothness condition for an approximating polynomial for $f$, and generalize this condition to partial functions defined on the $p$-biased hypercube. Lastly, we discuss other ways one may obtain such completions by perturbing the approximating polynomial.

\subsection{Discussion and further work}
In this work we introduced two new frameworks for proving non-superpolynomial quantum speedups and analyzed the relationship between partial functions and $S_n$. Furthermore, we showed several applications of the completion technique. We also introduced a new technique using the smoothness of an approximating polynomial to prove that no superpolynomial speedup exists. We list some open problems that would further the usefulness of our techniques. 

\begin{itemize}
    \item Completion complexity characterizes the possibility of superpolynomial speedups in comparison to deterministic query complexity.
    Is there a stronger characterization when comparing quantum to randomized query complexity? 
    \item Since the completion of an approximate degree captures the gap between deterministic and approximate degree -- and upper-bounds the gap between deterministic and quantum complexity -- but some functions exhibit superpolynomial quantum speedups, it follows that not all polynomials are extendable. What intrinsic properties of a polynomial determine whether it admits such a completion?
    \item For what classes of functions is a low degree approximation of the domain indicator function achievable? An exact characterization would help determine when the na\"ive completion is sufficient to prove no superpolynomial speedup.

    \item Theorem~\ref{thm:measures_no_speedup} directly implies that partial functions on domains with a large promise (i.e., at most $\poly(\criticalblocksensitivity{f})$ undefined inputs) do not attain superpolynomial speedup. Can one define a measure of \emph{connectedness} on the Boolean hypercube and show that on highly-connected domains one cannot obtain superpolynomial speedup?

    \item Are there more fine-grained versions of the promise measures discussed in~\Cref{sec:promise_measures} which allow us to show the non-existence of a superpolynomial speedup?

\end{itemize}

\subsection{Organization of the paper}

The rest of the paper is divided into four main sections, with an additional result deferred to the Appendix.
\Cref{sec:prilim} presents preliminaries, including query complexity, (traditional and a few new) complexity measures for partial functions, and key Fourier analysis tools.
\Cref{sec:promise_measures} develops the necessary lemmas and presents the proof of Result~\ref{res:promise_measure_speedup}.
Then \Cref{sec:other_section} discusses classical vs.\ quantum results for functions under symmetric promises. 
Finally, \Cref{sec:quantum_speedup_completions} introduces completion complexity, one of the core concepts of our framework and applies this framework to analyze quantum speedups through completions, covering different families of Boolean functions.

\subsection*{Acknowledgments}
We thank Srijita Kundu for many helpful discussions during the early phase of this project. We also thank Nengkun Yu and Kunal Marwaha for helpful discussions. This research was supported by US Department of Energy (grant no DE-SC0023179) and US National Science Foundation (award no 1954311).

%% file: poly_approx.tex

\begin{tikzpicture}
\begin{axis}[
    width=12cm,
    height=4cm,
    xmin=-0.2, xmax=12.6,
    ymin=-0.05, ymax=1.05,
    axis lines=none,
    xtick=\empty,
    ytick=\empty,
    clip=false,
    enlarge x limits=false,
]

\def\railcol{gray!40}
\def\gridcol{black!40}
\def\gridw{0.7pt}
\def\gap{0.06}

\addplot[very thick, \railcol] coordinates {(0,0) (12,0)};
\addplot[very thick, \railcol] coordinates {(0,1) (12,1)};

\foreach \x in {0,1,...,12}{
  \addplot[\gridcol, line width=\gridw] coordinates {(\x,0) (\x,1)};
}
\foreach \x in {2,9}{
  \addplot[\gridcol, line width=\gridw] coordinates {(\x,0) (\x,\gap)};
  \addplot[\gridcol, line width=\gridw] coordinates {(\x,\gap) (\x,1)};
}
\foreach \x in {5,8,12}{
  \addplot[\gridcol, line width=\gridw] coordinates {(\x,0) (\x,1-\gap)};
  \addplot[\gridcol, line width=\gridw] coordinates {(\x,1-\gap) (\x,1)};
}

\node[anchor=west] at (axis cs:12.12,1) {1};
\node[anchor=west] at (axis cs:12.12,0) {0};

\addplot[
  ultra thick,
  densely dotted,
  red!70!black,
  smooth,
  tension=0.35,
  restrict y to domain=0:1,
] coordinates {
  (0,0.10) (0.5,0.06) (1,0.03) (1.5,0.06) (2,0.00)
  (2.5,0.10) (3,0.25) (3.5,0.55) (4,0.85) (4.5,0.97) (5,1.00)
  (5.5,0.78) (6,0.42) (6.5,0.55) (7,0.82) (7.5,0.96) (8,1.00)
  (8.5,0.55) (9,0.00)
  (9.5,0.08) (10,0.15) (10.5,0.35) (11,0.55) (11.5,0.82) (12,1.00)
};

\addplot[
  very thick,
  black,
  domain=0:12,
  samples=4000,
] {
  0.03 + 0.94 * (
    min(0.995, max(0,
      1
      - 1/(1+exp(-18*(x-1.75)))
      + 1/(1+exp(-45*(x-3.70)))
      - 1/(1+exp(-45*(x-8.75)))
      + 1/(1+exp(-12*(x-10.75)))
    ))
  )
};

\addplot[
  only marks, mark=*,
  mark size=3.4pt,
  mark options={fill=orange!80!brown, draw=black, line width=0.35pt},
] coordinates {
  (2,0) (5,1) (8,1) (9,0) (12,1)
};

\node[anchor=west, text=orange!60!brown] at (axis cs:6,0.36) {$p(x)$};
\node[anchor=west, text=black!60!black]  at (axis cs:8.9,0.75) {$P(x)$};

\end{axis}
\end{tikzpicture}

%% file: prelims.tex
\section{Preliminaries}\label{sec:prilim}

For a positive integer $n$, let $[n]=\{1,2,\ldots,n\}$.
For $x\in\{0,1\}^n$ and $S\subseteq[n]$, we write $x_S$ for the restriction of $x$ to the coordinates in $S$.
For $i\in[n]$, let $x^{i}$ denote the string obtained from $x$ by flipping its $i$th bit; more generally, for a
set\footnote{We will use terms set and block interchangeably, as is common in query complexity literature.} $\mathrm{B}\subseteq[n]$, let $x^{\mathrm{B}}$ denote the string obtained by flipping all bits in $\mathrm{B}$.

Let $\dom \subseteq \lbrace 0,1 \rbrace^n$ and consider a function $f: \dom \rightarrow \lbrace 0,1 \rbrace$.
Alternatively, we can express the same function $f: \lbrace 0,1 \rbrace^n \rightarrow \lbrace{0,1, * \rbrace}$ with $f(x) = *$
if and only if $x \notin \dom$.
We will switch between these two representations depending on the context.
We also view a function $f$ either as $f:\Dom(f)\subseteq\{0,1\}^n\rightarrow\{0,1\}$ or as $f:\Dom(f)\subseteq\{\pm1\}^n\rightarrow\{-1,1\}$. These two viewpoints are equivalent. The complement of a domain is denoted $\domain(f)^c = \{x\in \{0,1\}^n : x\not\in \domain{f}\}$.

Given a string $x \in \lbrace 0,1\rbrace^n$, we use $\wt(x)$ to denote its Hamming weight.

\subsection{Query complexity}
\noindent\textbf{Deterministic Query Complexity.} A decision tree (or query algorithm) $\T$ computing $f$ is a binary tree where each internal node has exactly two children. Each internal node $v$ is labeled by a variable $x_v$, and each leaf $\ell$ is labeled by a bit $b_\ell \in \{0,1\}$. The computation on input $x$ proceeds as follows: It starts at the root of $\T$ and follows a path based on queries. At each internal node $v$, the variable $x_v$ is queried, and the computation moves to the left child if $x_v = 0$ and to the right child if $x_v = 1$. When a leaf $\ell$ is reached, the output is $\T(x) := b_\ell$, and the process terminates. The deterministic decision tree complexity of $f$, denoted $\D(f)$, is the depth of the shallowest decision tree computing $f$, equivalently stated as the worst-case number of queries made by the optimal decision tree.

\noindent\textbf{Randomized Query Complexity.} A randomized decision tree $\mathcal{T}$ is a probability distribution over deterministic decision trees $\T$. It is said to compute a Boolean function $f$ if, for every valid input $x$,
\[
\Pr_{\T \sim \mathcal{T}}[\T(x) \neq f(x)] \leq 1/3.
\]
The randomized decision tree complexity of $f$, denoted $\R(f)$, is the minimum depth $d$ for which there exists a randomized decision tree $\mathcal{T}$ supported on deterministic trees of depth at most $d$ that satisfies this condition.
Although the error bound of $1/3$ may seem arbitrary, it can be replaced by any constant strictly smaller than $1/2$ without affecting the complexity measure by more than a constant factor.

\noindent\textbf{Quantum Query Complexity.} A quantum oracle for an input $x$ is a unitary transformation $\QO_x:\mathbb{C}^{2^m} \to \mathbb{C}^{2^m}$ on a $2^m$-dimensional complex space, where $m \geq \lceil\log n \rceil + 1$. It acts on basis states of the form $|i, b, z\rangle$, where $i$ (an index in ${1,\dots,n}$) takes $\lceil\log n \rceil$ bits, $b$ takes 1 bit, and $z$ is the remaining workspace. The oracle applies the transformation:
\[
\QO_x |i, b, z\rangle = |i, b\oplus x_i, z\rangle.
\]
By linearity, this defines $\QO_x$ on all of $\mathbb{C}^{2^m}$.

A quantum query algorithm with $q$ queries is a unitary transformation
\[
\mathsf{A} = \mathsf{U}_q \QO_x \mathsf{U}_{q-1} \QO_x \dots \QO_x \mathsf{U}_1 \QO_x \mathsf{U}_0,
\]
where each $\mathsf{U}_i$ is a fixed unitary independent of $x$. The algorithm’s output is determined by measuring the first qubit of $\mathsf{A}|0^m\rangle$.

The algorithm computes $f$ with bounded error if, for every $x\in\Dom(f)$, its output equals $f(x)$ with probability at least $2/3$. The bounded-error quantum query complexity of $f$, denoted $\Q(f)$, is the minimum number of queries $q$ needed for such an algorithm.

For an overview of key results and open problems in query complexity, see \cite{aaronson2021open, ambainis2018understanding, buhrman2002complexity}.

\subsection{Definition of measures for partial functions}\label{sec:partial_function_measures}

Due to the nature of partial functions, which are undefined on certain inputs, one must be careful when defining complexity measures.
Indeed, when it comes to partial functions, even for standard combinatorial functions, there is no \emph{canonical choice} of complexity measure, and there are often multiple definitions in the literature~\cite{chakraborty2022certificate, anshu2020querytocommunicationliftingadversarybounds}.  
Ideally, we would want these measures to be \emph{meaningful} and to be related to some algorithmic measures, just like in the case of total functions.
For example, some combinatorial measures, such as approximate degree, adversary bound etc., of total functions lower bound quantum query complexity $\Q(f)$.
Therefore, a useful definition of  these measures in the case of partial functions should also lower bound $\Q(f)$.
Here, we state the definitions we will use throughout the paper, and show that they satisfy the requisite properties. Note that various measures for partial functions were already considered in the literature; see~\cite{anshu2020querytocommunicationliftingadversarybounds, chakraborty2022certificate, bun2022approximate} for reference.

\begin{definition}[Degree]
For a \emph{(possibly partial)} Boolean function $f$, its degree is defined as the minimum degree of a real polynomial $p$ that satisfies  
\[
p(x) = f(x) \quad \text{for all } x \in \mathrm{Dom}(f),
\]  
as well as  
\[
p(x) \in [0,1] \quad \text{for all } x \in \{0,1\}^n.
\]  
\end{definition}
If $f$ is a total function, then there is a unique multilinear polynomial $p$ satisfying the above conditions.
We can similarly define approximate degree.

\begin{definition}[Approximate degree]
For a (possibly partial) Boolean function $f$, its approximate degree measure, denoted by $\approxdegree_\epsilon(f)$, is the minimum degree of a real polynomial $p$ that approximates $f$ in the sense that  
\[
|p(x) - f(x)| \leq \epsilon \quad \text{for all } x \in \mathrm{Dom}(f).
\]  
Additionally, \( p(x) \) must satisfy  
\[
p(x) \in [0,1] \quad \text{for all } x \in \{0,1\}^n.
\]  
\end{definition}
As it is common in the literature, we will use the shorthand $\approxdegree$ to denote $\approxdegree_{\frac{1}{3}}$. Note that for both $\deg(f)$ and $\approxdegree(f)$, we are asserting that the polynomial $p(x)$ is bounded between the output values (as opposed to unbounded $p(x)$ outside of the domain) as this provides a stronger lower bound of $\Q(f)$.

In a seminal result, \cite{beals2001quantum} showed that one can lower bound quantum query complexity by the approximate degree, which leads to the following chain of inequalities.
\begin{lemma}[\cite{beals2001quantum}]\label{lem:degree_lower_bound}
    For every Boolean function $f$, we have, $\D(f) \geq \R(f) \geq \Q(f) = \Omega\left({\approxdegree(f)}\right)$.
\end{lemma}

The above results were initially proved for total functions, but as mentioned in \cite{anshu2020querytocommunicationliftingadversarybounds}, one can define these measures for partial functions, and the results also hold for them.

\begin{definition}[Partial assignment]
A partial assignment $\alpha$ is an element of $\{0,1,*\}^n$. We say that $\alpha$ is consistent with an input $x\in \{0,1\}^n$ if for all $i\in [n]$, either $\alpha_i=x_i$ or $\alpha_i=*$.
\emph{Support} of a partial assignment $\alpha$, denoted as $\mathrm{Support}(\alpha)$, is the set $\lbrace i: \alpha_i \neq * \rbrace$.
\end{definition}

\begin{definition}[Certificate, \cite{chakraborty2022certificate}]
    Let  $f:\{0,1\}^n \rightarrow \{0,1, *\}$, and $f(x)=b\in \{0,1\}$. Let $\alpha$ be a partial assignment that is consistent with $x$.
    We say that $\alpha$ is 
    \begin{itemize}
        \item A $b$-certificate (for $x$) if for any $x^\prime$ consistent with $\alpha$, $f(x^\prime)=b$.
        \item A $\{b,*\}$-certificate (for $x$) if for any $x^\prime$ consistent with $\alpha$, $f(x^\prime)\in \{b, *\}$.
    \end{itemize}
Moreover, we define
    \begin{itemize}
        \item $\bxcertificate{b}{f}$ as the size of the smallest $b$-certificate of $x$ and
        \item $\bxperpcertificate{b}{f}$ as the size of the smallest $\{b,*\}$-certificate of $x$.
    \end{itemize}
\end{definition}

Essentially, a $b$-certificate certifies that we are always within the promise and output $b$, while a $\{b,*\}$-certificate certifies that \emph{if} we are in the promise, we output $b$. We define the following measures.

\begin{definition}[Certificate measure, \cite{chakraborty2022certificate}]\label{def:certificate_perp}
    Given $f:\{0,1\}^n \rightarrow \{0,1,*\}$ and $b\in \{0,1\}$, we have that:
    \begin{align*}
        \bcertificate{b}{f} &= \max_{x\in f^{-1}(b)} \bxcertificate{b}{f} \\
        \perpcertificate{f}&= \max \{\bcertificate{0}{f}, \bcertificate{1}{f} \}\\
        \bperpcertificate{b}{f} &= \max_{x\in f^{-1}(b)} \bxperpcertificate{b}{f}\\
        \certificate{f} &= \max \{\bperpcertificate{0}{f}, \bperpcertificate{1}{f} \}.
    \end{align*}
\end{definition}

While there are multiple ways to define the certificate measure for partial function, the definition of \certificate{f} above appears the most natural as it lower-bounds deterministic query complexity.\footnote{Although the notation \bcertificate{}{f} has been canonically used in the literature for total Boolean functions to represent what we denote here as $\mathrm{C'}$, we follow the notational convention established in \cite{chakraborty2022certificate} for the certificate complexity of partial functions.}
By definition, we have that $\certificate{f} \leq \perpcertificate{f}$ for all $f$ and for any total function $f$, we have, $\certificate{f} = \perpcertificate{f}$.
One should note that, $\certificate{f}$ could be significantly smaller than $\perpcertificate{f}$. For example, the function $f$ defined on the $1$-slice, which asks whether the location $i\in[n]$ such that $x_i =1$ is less than $\frac{n}{2}$, maximally separates these measures with $\certificate{f} = O(1)$ and $\perpcertificate{f} = \Omega(n)$~\cite{chakraborty2022certificate}.

\begin{definition}[Block notation]
Let $x\in\{0,1\}^n$ and $B\subseteq[n]$. We let $x^B$ denote the string obtained from $x$
by flipping the bits in $B$. Equivalently,
\[
(x^B)_i \;=\;
\begin{cases}
x_i, & i\notin B,\\
1-x_i, & i\in B.
\end{cases}
\]
We refer to such a set $B$ as a \emph{block}.
If $B$ is a block of size one (say $B=\{j\}$), then we simply write $x^j$ (instead of $x^{\{j\}}$).
\end{definition}

Following the ideas of~\cite{chakraborty2022certificate}, we define two promise versions of sensitivity and block sensitivity with the main difference being how we process undefined inputs.

\begin{definition}[Sensitivity and block sensitivity]\label{def:sensitivity_perp}
    Given $f: \{ 0,1 \}^n \rightarrow \{0,1, *\}$ and $x\in \domain(f)$, let:
    \begin{itemize}
        \item \xsensitivity{x}{f} be number of variables $i \in [n]$ such that $f(x^i) = 1 - f(x)$.
        \item \xperpsensitivity{x}{f} be number of variables $i \in [n]$ such that $f(x^i)\in\{ 1-f(x), *\}$.
        \item \xblocksensitivity{x}{f} be the maximum number of disjoint blocks $B \subseteq [n]$ such that $f(x^B) = 1 -f(x)$.
        \item \xperpblocksensitivity{x}{f} be the maximum number of disjoint blocks $B \subseteq [n]$ such that $f(x^B)\in\{ 1-f(x), *\}$.
    \end{itemize}
    
    Then we define,
    \begin{align*}
        \sensitivity{f} & := \max_{x\in\domain(f)} \xsensitivity{x}{f}\\
        \perpsensitivity{f} &:= \max_{x\in\domain(f)} \xperpsensitivity{x}{f}\\
        \blocksensitivity{f} &:= \max_{x\in\domain(f)} \xblocksensitivity{x}{f}\\
        \perpblocksensitivity{f} &:= \max_{x\in\domain(f)} \xperpblocksensitivity{x}{f}.
    \end{align*}
    Note that the maximum is taken only over points in the domain, that is, over all those $x$ such that $f(x) \neq *$. 
\end{definition}

By definition, for any total function $f$, $\perpsensitivity{f}=\sensitivity{f}$ and $\perpblocksensitivity{f}=\blocksensitivity{f}$. However, these equalities need not hold for partial functions. 
The dictator function on the $1$-slice maximally separates $\sensitivity{f}$ and \blocksensitivity{f} from \perpsensitivity{f} and \perpblocksensitivity{f}, respectively.

\begin{observation}\label{obs:bs_s_c_basics}
    By Definitions~\ref{def:certificate_perp} and \ref{def:sensitivity_perp},  for all $f,b,x$, we get the following relations:
    \begin{align*}
        \xperpsensitivity{x}{f} &\leq \xperpblocksensitivity{x}{f} \leq \bxcertificate{f(x)}{f} \\
        \xsensitivity{x}{f} &\leq \xblocksensitivity{x}{f} \leq \bxperpcertificate{f(x)}{f} \\
        \xsensitivity{x}{f} &\leq \xperpsensitivity{x}{f}\\
        \xblocksensitivity{x}{f} &\leq  \xperpblocksensitivity{x}{f}.
    \end{align*}
\end{observation}

\begin{definition}[$\mathcal{A}$-sensitive blocks]
    Let $\mathcal{A} \subsetneq \{0,1,*\}$.
We say that a block $B$ is $\mathcal{A}$-sensitive (for an input $x$) if $f(x)\not \in \mathcal{A}$ and $f(x^B)\in \mathcal{A}$. In case that $\mathcal{A} = \{a\}$,  we simply say $a$-sensitive.
\end{definition}

Recall that we say a function $f^\prime$ is a completion of $f$ if $f'$ is total and $f(x)=f'(x)$ for all $x \in \dom(f)$.
\begin{definition}[Critical Block Sensitivity~\cite{huynh2012virtue}]
    Let $f: \{0,1\}^n \rightarrow \{0,1, *\}$.
    The critical block sensitivity of $f$ is
    \begin{align*}
        \criticalblocksensitivity{f} = \min_{f^\prime} \max_{x\in \domain(f)} \xblocksensitivity{x}{f^\prime}
    \end{align*}
    where $f^\prime$ is a completion of $f$.
\end{definition}

Furthermore, using critical block sensitivity $\criticalblocksensitivity{f}$, we define critical certificate,
\begin{definition}[Critical Certificate~\cite{anshu2020querytocommunicationliftingadversarybounds}]
    Let $f: \{0,1\}^n \rightarrow \{0,1, *\}$.
    Let $\mathcal{F} = \lbrace f' \text{ such that } f' \text{ is a }\\ \text{completion of } f \text{and } \max_{x\in\domain(f)} \xblocksensitivity{x}{f^\prime} = \criticalblocksensitivity{f}   \rbrace$.
    The critical certificate complexity of $f$ is,
    \begin{align*}
        \criticalcertificate{f} = \min_{f^\prime \in \mathcal{F}} \max_{x\in \domain(f)} \certificate{f^\prime,x}.
    \end{align*}
\end{definition}

By definition, for any total function $f$, we have, $\criticalblocksensitivity{f} = \blocksensitivity{f}$ and $\criticalcertificate{f} = \certificate{f}$. Lastly, the following relationship between $\blocksensitivity{f}, \criticalblocksensitivity{f}$ and $\approxdegree(f)$ holds for any Boolean function $f$.
\begin{lemma}[\cite{anshu2020querytocommunicationliftingadversarybounds}]\label{lem:bs_leq_approxdeg}
    For every partial Boolean function $f$,
    \begin{align*}
        \blocksensitivity{f} \leq \criticalblocksensitivity{f} =O( \approxdegree (f)^2).
    \end{align*}
\end{lemma}

It should be noted that the proof that $\blocksensitivity{f} = O(\approxdegree(f)^2)$ follows by constructing a polynomial which turns the blocks into variables and via symmetrization argues that the approximate degree must be large. In particular, this means that in \Cref{lem:bs_leq_approxdeg} we can not replace $\blocksensitivity{f}$ with $\perpblocksensitivity{f}$.

\subsection{Fourier analysis}\label{sec:fourier_analysis}

We review some basic Fourier analysis. Refer \cite{odonnell2021analysisbooleanfunctions} for a more detailed treatment of the subject. We note that $S \ni i$ refers to all sets $S\subseteq [n]$ that contain the index $i\in [n]$.

\begin{definition}[Discrete derivative]\label{def:discrete_derivative}
    Given a function $f: \{-1,1\}^{n}\rightarrow \mathbb{R}$, we define the discrete derivative of $f$ for all $i\in [n]$ as
    \begin{align*}
        D_{i}f(x)=\frac{f(x)-f(x^{i})}{2}.
    \end{align*}
\end{definition}

\begin{definition}[Influence]\label{def:influence}
    Given a function $f(x): \{-1,1\}^n \rightarrow \mathbb{R}$, we define the influence of the bit $i\in [n]$ as
    \begin{align*}
        \Inf_i[f] = \sum_{S\ni i} \hat{f}(S)^2.
    \end{align*}
\end{definition}

\begin{definition}[$p$-biased hypercube]
Let $\pi_p$ denote the distribution over $\{ -1,1 \}$ with $\pi_p(-1)=p$.
With slight abuse of notation, we extend the notation to $x \in \lbrace -1, 1\rbrace^n$ : $\pi_p(x) = \prod\limits_{i \in [n]}\pi_p (x_i)$.
We will denote this distribution as $\lbrace -1,1\rbrace^n_p$.
\end{definition}

\begin{definition}\label{def:p_biased_phi} Let $p \in (0,1)$.
    Define the function $\phi: \{-1,1\} \rightarrow \mathbb{R}$ by
    \[
    \phi(x_{i})=\frac{x_{i}-\mu}{\sigma}
    \]
    where,
    \[
\mu=\mathop{\mathbb{E}}\limits_{x_{i}\sim \pi_{p}}[x_{i}]=1-2p \quad \text{ and }  \hskip .3cm \sigma=\mathop{\mathrm{stddev}}\limits_{x_{i}\sim\pi_{p}}[x_{i}]=2\sqrt{p(1-p)}.
    \]
\end{definition}

\begin{definition}[Fourier basis for the $p$-biased hypercube]\label{def:p_biased_basis}
    We define the Fourier basis for functions on the $p$-biased hypercube as $\{\phi_{S}\}_{S\subseteq[n]}$ where
    \[
    \phi_{S}(x)=\prod_{i\in S}\phi(x_{i}).
    \]
    The Fourier coefficients are defined as
    \[
    \hat{f}(S)=\mathop{\mathbb{E}}\limits_{x\sim\pi_{p}^{n}}[f(x)\phi_{S}(x)]
    \]
    and so we have the biased Fourier expansion
    \[
    f(x)=\sum_{S\subseteq[n]}\hat{f}(S)\phi_{S}(x).
    \]
\end{definition}

We also note the following property of smooth functions.

\begin{definition}[Lipschitz continuous function between $x$ and $y$]
    A function $f:\mathbb{R}^{n} \rightarrow \mathbb{R}$ is \textit{Lipschitz continuous} between two points $x$ and $y$ if $|f(x)-f(y)|\leq L|x-y|$ where $L>0$ is called the Lipschitz constant. We say that a function $f$ with Lipschitz constant $L$ is $L$-Lipschitz continuous.
\end{definition}

\begin{lemma}[Lemma 25, \cite{Gily_n_2019}]
\label{lem:chebyshev_sign}
    For all $\delta>0, \epsilon\in (0,1/2)$ there exists an efficiently computable odd polynomial $P\in \mathbb{R}[x]$ of degree $n=O\left( \frac{\log(1/\epsilon)}{\delta} \right)$, such that
    \begin{itemize}
        \item for all $x\in [-2,2]: |P(x)|\leq 1$ and,
        \item for all $x\in [-2,2]\setminus (-\delta,\delta): |P(x)-sign(x)|\leq \epsilon$.
    \end{itemize}
\end{lemma}

Let $\mathcal{P}_n$ denote the set of all multilinear polynomials in $n$ variables that are bounded between $[-1,1]$ on $\{-1,1\}^n$. 
Similarly, let $\mathcal{P}_{n,d} \subseteq \mathcal{P}_n$ denote the set of mulilinear polynomials with degree at most $d$.
\begin{definition}[Sparsity of a multilinear polynomial]
    The sparsity of a multilinear polynomial $p$, denoted $\spar(p)$, is defined as the number of non-zero coefficients in the exact Fourier representation of $p$. That is, $\spar(p)=|\{ S \subseteq [n] : \hat{p}(S)\neq 0\}|$. Furthermore, for an index $i\in [n]$, $\spar_i(p) = \abs{\{ S \subseteq [n] : i\in S \text{ and }\hat{p}(S)\neq 0 \}}$.
\end{definition}

\begin{definition}[Sparsity of partial functions]
    Let $f: \domain{f} \rightarrow \{-1,1\} \text{  with }  \domain{f}\subseteq\{-1,1\}^{n}$.
    Then, 
    \[ \spar(f) = \min_{p \in \mathcal{P}_{n,\deg(f)}} \lbrace \spar(p):\ p(x)=f(x) \text{ for all } x \in \dom(f)  \rbrace. \]
Similarly, for $\epsilon>0$, we define
 \[ \approxspar_\epsilon(f) = \min_{p \in \mathcal{P}_{n,\approxdegree_\epsilon(f)}} \lbrace \spar(p):\ \vert p(x) - f(x) \vert \leq \epsilon \text{ for all } x \in \dom(f).  \rbrace \]
        For the ease of notation, we denote $\approxspar_\frac{1}{3}$ as $\approxspar$.
\end{definition}

Furthermore, we will be discussing functions restriced to domains with the same hamming weight, called the slice. Formally, we define the hamming slice $\slice{n,k}$ at weight $k$ and its union over $K\subseteq [n]$ slices such that $K\neq \emptyset$ as,
\begin{align*}
    \slice{n,k} &\coloneq \{x\in \{0,1\}^n: \wt(x)=k\},\\
    \unionSlice{K} &\coloneq \bigcup_{k\in K} \slice{n,k}.
\end{align*}
We will use the following results on slices.

\begin{theorem}[Theorem 3.1 in~\cite{DFLLV21}]\label{thm:slice_bs_d_bound}
    For every partial function $f:\slice{n,k} \to \{0,1\}$, if $4\bs(f)^2 \leq \min\{k,n-k\}$, then $\D(f)\leq 4\bs(f)^4$.
\end{theorem}

\begin{lemma}[Theorem 14 in~\cite{bendavid2014structurepromisesquantumspeedups}]\label{lem:slice_c_bs}
    For every partial function $f:\slice{n,k} \to \{0,1\}$, $\certificate{f} \leq 3\bs(f)$.
\end{lemma}

%% file: perp.tex
\section{Promise measures}\label{sec:promise_measures}

A central question in quantum query complexity is whether superpolynomial quantum speedups require a nontrivial promise. This has been formalized in several settings: symmetric and permutation-symmetric partial functions admit at most polynomial quantum advantages, and similar no-speedup results are known for functions on structured domains such as slices~\cite{Need_For_Structure,chailloux2019note,bendavid2014structurepromisesquantumspeedups}. This section is motivated by the goal of identifying a more local obstruction that does not rely on global symmetry. The key point is that for partial functions, local perturbations may leave the domain rather than reach an opposite-labeled input. By analyzing the promise-aware versions of the usual combinatorial measures we defined in~\Cref{sec:partial_function_measures} and studying when they still satisfy total-function-type relations, this section aims to identify conditions for superpolynomial quantum speedups.

\begin{figure}[ht]
    \centering
    \input{measure_relationships}
    \caption{Relationships between the measures discussed.  Measures introduced in this work are in \textcolor{BurntOrange}{orange} and the relationships proved in this work are in \textcolor{ForestGreen}{green}.}
    \label{fig:measure_relationships}
\end{figure}
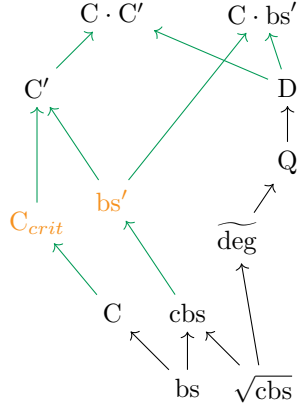

\begin{lemma}\label{lem:cf_leq_bs_sperp}
    For any partial function $f:\domain(f)\rightarrow \{0,1\}$, we have, 
    \begin{align*}
        \certificate{f} \leq \blocksensitivity{f}\perpsensitivity{f}.
    \end{align*}
\end{lemma}
\begin{proof}
    Consider an arbitrary input $x\in \domain(f)$.
    Without loss of generality, assume $f(x)=1$. Let $b= \xblocksensitivity{x}{f}$ and $B_1,\ldots, B_b$ be disjoint minimal blocks that achieve this block sensitivity for $x$.
    Since $B_1, \ldots, B_b$ are sensitive blocks, we get that $f(x^{B_i}) = 0$ for all $i\in[b]$.
    Consider the partial assignment $C$ which agrees with $x$ on $B_1,\ldots,B_b$ and has $*$ everywhere else.
    
    We claim that $C$ is a $\{1,*\}$-certificate.
    For contradiction, let's assume this is not the case.  
    Then, there exists some $y\in \domain(f)$ that is consistent with $C$, but $f(y)=0$.
    Let $B$ be the set of indices where $x$ and $y$ differ, that is, $x^B=y$.
    By construction, for all $i$, $B_i \subseteq \mathrm{Support}(C)$.
    Since $C$ is consistent with both $x$ and $y$, we get that $x$ and $y$ must agree on $B_i$. 
    In particular, $B$ must be disjoint from all $B_i$. 
    Moreover $f(x^B)=0$, making $B$ a sensitive block of $x$. 
    Thus, $B_1 \ldots, B_b, B$ form disjoint sensitive blocks for $x$, contradicting the assumption $b=\xblocksensitivity{x}{f}$.
  Thus $C$ with $\mathrm{Support}(C) = \bigsqcup B_i$ must be a $\{1,*\}$-certificate.
  
    Next, we claim that the size of any \emph{minimal} $B_i$ can be at most $\perpsensitivity{f}$. 
    This can be seen as follows.
    Consider any $j\in B_i$.
    Then, $f(x^{B_i\setminus \{j\}}) \neq 0$; otherwise,  $B_i\setminus \{j\}$ becomes a sensitive block of $x$, contradicting the minimality of $B_i$. 
    Thus, $f(x^{B_i\setminus \{j\}}) \in \{ 1, *\}$.
    Moreover, since $B_i$ is a sensitive block of $x$, we have $f(x^{B_i})=0$. 
    Hence, $j$ is a $\{1,*\}$-sensitive index for $x^{B_i}$. Since each index $j \in B_i$ is sensitive for  $x^{B_i}$, we get, $\abs{B_i} \leq \xperpsensitivity{x^{B_i}}{f} \leq \perpsensitivity{f}$. 
    
    Combining the two claims, we get that the size of certificate for $x$ is at most $\blocksensitivity{f}\perpsensitivity{f}$.
    Since we did this for arbitrary $x \in \domain(f)$, this completes the proof of the lemma.
\end{proof}

Note that in the above lemma, we showed that $\vert B_i \vert \leq s'(f)$.
For a total function $f$, it also holds that $\abs{B_i} \leq \sensitivity{f}$. This does not transfer to partial functions.  Moreover, $\{1-f(x),*\}$-sensitive blocks $B^*$ can neither be bounded by $s$ or by $s'$. For example, consider the constant $0$ function whose domain is $\{0,1\}^n\setminus \{0^n\}$.

\begin{lemma}\label{lemma_df_leq_c1_bs}
    For any partial function $f$ the following inequalities hold:
    \begin{enumerate}
        \item $ \deterministicquery{f} \leq \min\{\bperpcertificate{1}{f},\bperpcertificate{0}{f}\} \cdot  \perpblocksensitivity{f} $\label{eq:d_leq_c_perpbs}
        \item $\deterministicquery{f} \leq \min\{\bcertificate{1}{f},\bcertificate{0}{f}\}  \cdot \blocksensitivity{f}$ \label{eq:d_leq_perpc_bs}
        \item   $\deterministicquery{f} \leq \criticalcertificate{f} \cdot \criticalblocksensitivity{f}$. \label{eq:d_leq_crit_cbs}
    \end{enumerate}
 
\end{lemma}
\begin{proof}
    All statements follow a similar argument to the result $\deterministicquery{f} \leq \certificate{f} \blocksensitivity{f}$ for total functions~\cite{beals2001quantum}. We show $\deterministicquery{f}\leq \bperpcertificate{1}{f}\perpblocksensitivity{f}$ and  $\deterministicquery{f}\leq \bperpcertificate{0}{f}\perpblocksensitivity{f}$ follows by symmetry, proving \Cref{eq:d_leq_c_perpbs}.
    The proof holds due to the following algorithm. 
    
    \begin{algorithm}\caption{}
        Let $p$ be the set of queried indices so far, initialized as $p=\emptyset$. Repeat the following for $b = \perpblocksensitivity{f}$ rounds:
    \begin{enumerate}
        \item Let $i\in [b]$ be the current round. Let $c_i$ be a $\{1,*\}$-certificate of size at most $\bperpcertificate{1}{f}$ for an input $x\in \domain(f)$ such that $f(x)=1$ which agrees with $p$.
        \item If no such $c_i$ exists, output $0$.
        \item Query the indices in $c_i$. If they agree with $c_i$, output $1$. Otherwise, add all the queried indices and their results to $p$.
    \end{enumerate}
    Suppose the algorithm did not terminate. As will be shown below, for all $x,y\in \domain(f)$ which agree with $p$, $f(x)=f(y)$. Therefore, based on $p$, we may output $f(x)$ (this may be hard-coded into the algorithm).
    \end{algorithm}

    Let us explain the correctness of the algorithm. If we terminate during the $b$ rounds, we output the correct result. To show that after $b$ rounds all inputs which agree with $p$ have the same output, assume the statement is false and there are two inputs $x,y\in \domain(f)$ which agree with $p$ and $f(x)\neq f(y)$. Without loss of generality, assume $f(y)=0$. Let $B_{b+1}$ be the block such that $x^{B_{b+1}}=y$. For each $c_i$, we may define a corresponding block $B_i$ as the set of indices on which $y$ and $c_i$ disagree. Each $B_i$ is non-empty, disjoint from all other $B_j$, and $f(y^{B_i}) \in \{1, *\}$, meaning that $B_1,\dots B_{b+1}$ form a valid set of blocks. However, this would contradict the assumption that the block sensitivity $\perpblocksensitivity{f}=b$. Thus, the algorithm must output the correct value. Items~\ref{eq:d_leq_perpc_bs} and~\ref{eq:d_leq_crit_cbs} follow analogously.
\end{proof}

Alternatively, we may bound $\deterministicquery{f}$ using $\certificate{f}$ and $\perpcertificate{f}$.

\begin{lemma}\label{lem:df_c_cperp}
    For any partial function $f:\domain(f)\rightarrow \{0,1\}$,
    \begin{align*}
        \deterministicquery{f} \leq \min\{\bcertificate{1}{f} \bperpcertificate{0}{f}, \bcertificate{0}{f} \bperpcertificate{1}{f}\}.
    \end{align*}
\end{lemma}
\begin{proof}
Let us prove $\deterministicquery{f} \leq \bcertificate{1}{f} \bperpcertificate{0}{f}$, the rest of the statement follows by symmetry. Fix some $x\in \domain(f)$. Let us describe a deterministic algorithm which outputs $f(x)$ using at most $\bcertificate{1}{f} \bperpcertificate{0}{f}$ queries.

\begin{algorithm}\caption{}
    Let $p$ be a partial assignment which is initially empty which records all bits of $x$ queried so far. Additionally, assume there exists a fixed ordering of all $1$-certificates of $f$ of size at most $\bcertificate{1}{f}$. Repeat the following for $\bperpcertificate{0}{f}$ rounds:
\begin{enumerate}
    \item If $p$ contains a $\{0,*\}$-certificate, output $0$.
    \item Otherwise, choose the first $1$-certificate $c$ consistent with $p$, query its indices and add them to $p$.
    \item If the updated $p$ agrees with $c$, output $1$.
\end{enumerate}
If the loop does not terminate early, output $0$ if $p$ contains a $\{0,*\}$-certificate, and $1$ otherwise.
\end{algorithm}

Let us analyze the algorithm. Since the size of the certificate $c$ queried each round is bounded by $\bcertificate{1}{f}$, the query complexity is bounded by $\bcertificate{1}{f} \bperpcertificate{0}{f}$.  Let us argue for correctness. If we terminate in Step 1, we are correct as since $x\in \domain(f)$, a $\{0,*\}$-certificate implies that $f(x)=0$. On the other hand, if we continue to Step 2, we are guaranteed to find some $1$-certificate $c$ which agrees with $p$ as otherwise, we could create a $\{0,*\}$-certificate. Of course, if $p$ agrees with $c$, $f(x)=1$ and we output the correct value.

Lastly, suppose $f(x)=0$ and the loop terminates without an output. Let $c_0$ be the minimum $\{0,*\}$-certificate for $x$, meaning that $\abs{c_0}\leq \bperpcertificate{0}{f}$. Additionally, notice that for each $c$ in Step 2, there must exist at least $1$ index on which $c$ and $c_0$ disagree as both cannot be true at the same time. Since $c$ and $c_0$ both agree with $p$ before Step 2, this index has not yet been queried. Therefore, at each round of the algorithm, we must query at least one new index in $c_0$. Therefore, after $\bperpcertificate{0}{f}$ rounds, we must have queried all of its indices, so $p$ contains $c_0$ and the final output is $0$. If $f(x)=1$, then $p$ cannot contain a $\{0,*\}$-certificate, so the final output is $1$.
\end{proof}

Next, we find that \criticalblocksensitivity{f} sits between \blocksensitivity{f} and \perpblocksensitivity{f}.

\begin{lemma}
   $
        \blocksensitivity{f} \leq \criticalblocksensitivity{f} \leq \perpblocksensitivity{f}$.

\end{lemma}
\begin{proof}
    By definition, the completion of $f$ must have no fewer sensitive blocks than $f$ itself, so $\blocksensitivity{f} \leq \criticalblocksensitivity{f}$. On the other hand, consider the completion $f^\prime$ of $f$ which achieves block sensitivity $c = \criticalblocksensitivity{f}$ and let $B_1,\dots B_c$ be its blocks. Then for that input $x$, each block either points to an input in the domain or not. Regardless, each block may be used as a block for \perpblocksensitivity{f}. Thus, $\criticalblocksensitivity{f} \leq \perpblocksensitivity{f}$.
\end{proof}

Lastly, we prove a statement which shows how a polynomial relationship between various measures discussed in this section implies the impossibility of a quantum speedup.

\begin{theorem}\label{thm:measures_no_speedup}
    Let $f$ be a partial function such that \emph{at least} one of the following conditions is satisfied:
    \begin{enumerate}
        \item\label{itm:blocksensitivity_poly} $\polyequal{\criticalblocksensitivity{f}}{\perpblocksensitivity{f}}$
        \item \label{itm:cbs_ccrit} $\polyequal{\criticalblocksensitivity{f}}{\criticalcertificate{f}}$
        \item \label{itm:perp_cert_bs} $\min\{\bcertificate{1}{f},\bcertificate{0}{f}\} \polyleq \criticalblocksensitivity{f}$.
        \item \label{itm:perpcertificate_bs_perps_poly} $\polyequal{\certificate{f}}{\perpcertificate{f}}$ and $\polyequal{\blocksensitivity{f}}{\perpsensitivity{f}}$.
    \end{enumerate}
    Then, $\polyequal{\deterministicquery{f}}{\Q(f)}$.
\end{theorem}
\begin{proof}
    Let us begin with Item~\ref{itm:blocksensitivity_poly}. Let $a$ be some constant such that $\perpblocksensitivity{f} \leq \criticalblocksensitivity{f}^a$. Then,
    \begin{flalign*}
        &&\deterministicquery{f} &\leq \certificate{f} \perpblocksensitivity{f} && \mbox{(from~\Cref{eq:d_leq_c_perpbs} in Lemma~\ref{lemma_df_leq_c1_bs})}\\
        &&&\leq\blocksensitivity{f} \perpsensitivity{f} \perpblocksensitivity{f} &&\mbox{(from Lemma~\ref{lem:cf_leq_bs_sperp})}\\
        && &\leq \perpblocksensitivity{f}^3 &&\mbox{(from~\Cref{obs:bs_s_c_basics})}\\
        &&&\leq \criticalblocksensitivity{f}^{3a}\\
        &&&= O(\Q(f)^{6a}). &&\mbox{(from~\Cref{lem:bs_leq_approxdeg})}
    \end{flalign*}
    This completes the proof that Item~\ref{itm:blocksensitivity_poly} is sufficient to ensure that $\polyequal{\deterministicquery{f}}{\Q(f)}$.
    The proof that
    Items~\ref{itm:cbs_ccrit} and~\ref{itm:perp_cert_bs} are also sufficient follows similarly (due to Items~\ref{eq:d_leq_crit_cbs} and~\ref{eq:d_leq_perpc_bs} in Lemma~\ref{lemma_df_leq_c1_bs} and~\Cref{lem:bs_leq_approxdeg}, respectively).

    Item~\ref{itm:perpcertificate_bs_perps_poly}. Let $b$ be the smallest constant such that $\blocksensitivity{f} \leq \perpsensitivity{f}^b$ and $\perpsensitivity{f} \leq \blocksensitivity{f}^b$, while $c$ is the smallest constant such that $\perpcertificate{f} \leq \certificate{f}^c$. Then,
    \begin{flalign*}
        &&\D(f) &\leq \certificate{f} \perpcertificate{f} &\mbox{(from Lemma~\ref{lem:df_c_cperp})}\\
        &&&\leq \certificate{f}^{c+1} \\
        &&&\leq \blocksensitivity{f}^{(b+1)(c+1)} &\mbox{(from Lemma~\ref{lem:cf_leq_bs_sperp})}\\
        &&&= O(\Q(f)^{8bc}). &\mbox{(from~\Cref{lem:bs_leq_approxdeg})} &\qedhere
    \end{flalign*}
\end{proof}

Theorem~\ref{thm:measures_no_speedup} implies that several structured domains cannot attain a speedup. For example, Item~\ref{itm:blocksensitivity_poly} implies that $\polyequal{\blocksensitivity{f}}{\perpblocksensitivity{f}}$ leads to no superpolynomial quantum speedups.
It also implies that for a domain to attain speedup, there must exist at least one input $x$ with many disjoint blocks $B$ such that $x^B\not\in \domain(f)$. Alternatively, if the graph induced by adjacent inputs in the domain is \say{well-connected} such that the degree of each node is closer to $n$ than the block sensitivity of the function, it cannot attain speedup.

Item~\ref{itm:cbs_ccrit} of Theorem~\ref{thm:measures_no_speedup} is consistent with the discussion in \Cref{sec:completion}: if a partial function admits a completion that does not increase a relevant parameter, then the corresponding deterministic and quantum query complexities are polynomially related.

As a simple example, fix some input $x^* \in \{0,1\}^n$ and consider a partial Boolean function $f$ with domain
$
\domain(f) = \{x : \forall i \in [n],\ x_i \leq x_i^* \}.
$
On this domain, define a completion $g$ by
\[
g(x) = f(x \land x^*), \text{ where $(x \land x^*)_i = x_i \land x_i^*$.}
\]
 This completion preserves the relevant measures: every minimal sensitive block and every certificate depends only on coordinates $i$ such that $x_i^* = 1$. In particular,
\[
\max_{x \in \domain(f)} \xblocksensitivity{x}{g}
= \criticalblocksensitivity{f}
= \blocksensitivity{f},
\qquad\text{and}\qquad
\criticalcertificate{f} = \certificate{f}.
\]

This argument extends to domains of the form $\bigcup_i d(x_i^*)$, provided that
\[
\max_{i,j} \wt(x_i^* - x_j^*) = O(1),
\]
as in this case, the coordinates on which the various $x_i^*$ differ can be queried to identify the subdomain containing the input.

Interestingly, the conclusion of Theorem~\ref{thm:measures_no_speedup} does not follow from $\polyequal{\certificate{f}}{\perpcertificate{f}}$, meaning the second condition in Item~\ref{itm:perpcertificate_bs_perps_poly} is necessary. This is due to the Deutsch-Jozsa algorithm which solves the problem of whether some input $x\in \{0,1\}^n$ is constant (i.e., $\abs{x}=0$ or $n$) or perfectly balanced ($\abs{x}=n/2$)~\cite{deutsch_jozsa}. In this case, all three of $\deterministicquery{f}, \certificate{f}$ and $\perpcertificate{f}$ are $O(n)$, while $\Q(f)=1$.

%% file: measure_relationships.tex
\begin{tikzpicture}
    \node (bs) at (0,0) {$\mathrm{bs}$};
    \node (cbs) at (0,1) {$\mathrm{cbs}$};
    \node (cbssquare) at (1,0) {$\sqrt{\mathrm{cbs}}$};
    \node (bsperp) at (-1,2.5) {$\color{BurntOrange}\mathrm{bs}^\prime$};
    \node (cperp) at (-2,4) {$\mathrm{C}^\prime$};
    \node (c) at (-1,1) {$\mathrm{C}$};
    \node (deg) at (0.66,2) {$\widetilde{\mathrm{deg}}$};
    \node (q) at (1.33,3) {$\Q$};
    \node (d) at (1.33,4) {$\D$};
    \node (ccrit) at (-2,2.2) {$\color{BurntOrange}\mathrm{C}_{crit}$};
    \node (csquare) at (-1,5) {$\mathrm{C}\cdot \mathrm{C}^\prime$};
    \node (bsperp_c) at (1,5) {$\mathrm{C}\cdot \mathrm{bs}^\prime$};

     \foreach \source/\target in {bs/cbs, bs/c, cbssquare/deg, deg/q, q/d, cbssquare/cbs}
    \draw[->] (\source) -- (\target);
    \foreach \source/\target in {cbs/bsperp, bsperp/cperp, c/ccrit, ccrit/cperp, cperp/csquare, d/csquare, bsperp/bsperp_c, d/bsperp_c}
    \draw[->, ForestGreen] (\source) -- (\target);
\end{tikzpicture}

%% file: other_results.tex
\section{Functions with $S_n$ symmetry}\label{sec:other_section}
In this section, we look at functions that have some invariance under the symmetric group $S_n$. 
In the first part, we consider functions having the highest order of $S_n$ symmetry (Definition~\ref{def:symmetric_function}).
In the second part, we will consider functions defined on a single orbit of $S_n$, also known as the Boolean slice.
However, in this case, the function may be arbitrarily defined on the slice.
\subsection{Symmetric functions}
In this part, we will study functions that are symmetric as defined below.
\begin{definition}[Symmetric function] \label{def:symmetric_function}
Let $\dom(f)  \subseteq \lbrace 0,1 \rbrace^n$.
We say that a function $f: \domain(f) \rightarrow \lbrace 0,1\rbrace$ is symmetric under $S_n$ (or simply symmetric) if for all $\sigma \in S_n$ and for all $x \in \lbrace 0,1 \rbrace^n$:
\begin{enumerate}
    \item $x \in \dom(f) $ if and only if $\sigma(x) \in \dom(f) $.
    \item $f(x) = f(\sigma(x))$ for $x\in\dom(f)$.
\end{enumerate}
\end{definition}

Note that for a symmetric function $f$, its evaluation at a point $x$ depends only on $\wt(x)$.

Our main focus will be on characterizing various measures such as $\D$ and $\Q$ for partial symmetric functions.
This problem has already been considered for $\R$ and $\Q$ in~\cite{chailloux2019note}, proving that there is no superpolynomial speedup between the two quantities. Furthermore, ~\cite{podder2025fine} analyzed these quantities and proved that $\Q$ can be written as a certain maximization problem, whereas our expression is free from such optimization and has a much simpler proof technique.
Specifically, we give a characterization of various measures (up-to a polynomial factor) using $\gap$, a parameter for symmetric functions we define below. Additionally, we also consider the relationship between $\approxdegree$ and $\D$, of which the former has been characterized for total symmetric Boolean functions by~\cite{paturi1992degree}, and the latter is trivial; for total symmetric functions as it is $\Theta(n)$ for non-constant functions~\cite{buhrman2002complexity}. Furthermore, our results exactly characterize when symmetric functions attain a superpolynomial separation between $\Q$ and $\D$ (or $\approxdegree$ and $\D$).
As an added result, our proof also shows that $\approxdegree$ and $\Q$ are polynomially related for every such function, hence the polynomial method is tight for partial symmetric functions. Let us formally define $\gap$.

\begin{definition}[Gap of a symmetric function] For any nonconstant partial symmetric function $f:\domain(f)\rightarrow \{0,1\}$, define $\gap(f)$ as
    \[ \gap(f) = \min_{x,y\in\dom(f): f(x) \neq f(y)} d_H(x,y) =\footnote{This equality need not be true for an arbitrary partial function, but holds for symmetric functions.}  \min_{x,y\in\dom(f): f(x) \neq f(y)} \vert \wt(x)-\wt(y) \vert. \]
\end{definition}

First, consider the following standard fact.

\begin{fact}[\cite{Shalev}] \label{fact: symmetrization_2}
Let $p : \mathbb{R} \to \mathbb{R}$ be a real polynomial. Let $n \in \mathbb{N}$ be a positive integer, and suppose that
$p(x) \in [0,1] \text{ for all } x \in \{0,1,\ldots,n\}.$
Then
\[
\deg(p) \ge \sqrt{\frac{n c}{1 + c}},
\text{\hspace{1cm} where }
c = \max_{x \in [0,n]} \left| p'(x) \right|.
\]
\end{fact}

We will use the standard symmetrization of a polynomial.
Given a polynomial $p:\{0,1\}^n \to \mathbb{R}$, define its symmetrization $p_0:\{0,1,\dots,n\}\to \mathbb{R}$ by
\[
p_0(t) \;=\; \mathbb{E}_{x: \wt(x)=t}\,[\,p(x)\,].
\]
That is, $p_0(t)$ is the average value of $p(x)$ over all inputs of Hamming weight $t$.
In this case, we say that $p_0$ is the symmetrized version of $p$.

Now, let us characterize $\approxdegree, \Q$ and $\R$ using $\gap$.

\begin{lemma} \label{lem:Q_for_sym}
    Let $f$ be a nonconstant symmetric function. 
     Then, 
     
     $\approxdegree(f) \polyeq \Q(f) \polyeq \R(f) \polyeq \frac{n}{\gap(f)}$.
\end{lemma}
\begin{proof}
We will show this in two parts. 
First, $\frac{n}{\gap(f)} \polyleq \approxdegree(f)$.
Our proof will follow the standard symmetrization technique.
Consider $x,y$ that achieve the minimization of $\gap(f)$.
That is, let $x$ and $y$ be such that  $f(x)=0$, $f(y)=1$ and $\vert \wt(x) - \wt(y) \vert =\gap(f)$.
Let $p(x)$ be an optimal polynomial that approximates $f$.
That is, 
\begin{enumerate}
\item  $p(x) \in [0,1]$  for all $x$.
    \item $\vert p(x) - f(x) \vert \leq \frac{1}{3}$ for all $x \in \domain{f}$.
    \item  $\approxdegree(f) = \deg(p)$.
\end{enumerate}
Since $f$ is symmetric, $f(x')=0$ for all $x'$ such that $\wt(x')=\wt(x)$.
Thus, $p(x') \leq \frac{1}{3}$. Similarly, $p(y') \geq \frac{2}{3}$ for all $y'$ that have the same Hamming weight as $y$.
Let $p_0$ be the symmetrized version of $p$, then $p_0(\wt(x)) \leq \frac{1}{3}$ and $p_0(\wt(y)) \geq \frac{2}{3}$. 
Thus, $\max_{x \in [0,n]} \vert p_0^\prime (x)\vert \geq \frac{1}{3 \vert \wt(x) - \wt(y)\vert} = \frac{1}{3 \gap(f)}$.
Using Fact~\ref{fact: symmetrization_2}, $\deg(p_0) \geq \sqrt{\frac{ n }{3 \gap(f)+1}}$.
Thus, the standard symmetrization arguments gives us that $\approxdegree(f) = \deg(p)  \geq \deg(p_0) \geq \sqrt{\frac{n }{3 \gap(f)+1}}$.
 
For the other direction, we will give a randomized algorithm which outputs $f(x)$ with query complexity $t= O \left(\frac{n^2}{\gap(f)^2}\right)$. As $f(x)$ depends only on the Hamming weight $\mathrm{wt}(x)$, it suffices to estimate the value of $\wt(x)$ by taking uniformly random samples of indices of the input. Letting $\widehat{w}$ denote the weight estimate, since the hamming weights of any two conflicting inputs (i.e., $x,y \in \domain{f}$ with $f(x) \neq f(y)$) are at least $\gap(f)$ apart, we want to bound the probability $\widehat{w}$ differs by at least $\frac{1}{2} \cdot \gap(f)$. Suppose the randomized algorithm uniformly samples $t$ coordinates with replacement and queries them, meaning that $\widehat{w}$ is the number of $1$'s among these samples, scaled by $n/t$. By a Chernoff bound, after $t$ samples, the estimate satisfies
\[
\Pr\!\left[\, |\widehat{w} - \mathrm{wt}(x)| \geq \tfrac{1}{2}\,\gap(f) \,\right] \le 
e^{-\Omega(t\,(\gap(f)/n)^2)}.
\]
With $\sqrt t = \Theta\left(\frac{n}{\gap(f)}\right)$, we can make the probability a small enough constant,
therefore, $\R(f) \leq t = \Theta\left(\frac{n^2} {(\gap(f))^2}\right)$.
Thus all these measures are polynomially related. This completes the proof of the lemma.
\end{proof}

Lastly, we evaluate $\D$ using $\gap$.

\begin{lemma} \label{lem:d_for_sym}
    Let $f$ be a nonconstant symmetric function. 
    Then, $\D(f) = n-\gap(f) + 1$.
\end{lemma}
\begin{proof}
    Let $x,y$ be points that achieve the $\gap(f)$. Without loss of generality, assume that $\wt(x) \leq \wt(y)$. 
    Define the domain $\mathcal{D}^{*} = \lbrace z \ :\ z=\sigma(x) \text{ or } z=\sigma(y),\ \sigma\in S_n \rbrace \subseteq \domain{f}$ and define $f^{*}: \mathcal{D}^* \rightarrow \lbrace 0,1 \rbrace$ with $f^*(x)=f(x)$ for $x \in \mathcal{D}^*$.
    Since $f^{*}$ is a subfunction of $f$, any algorithm that computes $f$ must also compute $f^{*}$. 
    Hence, $\D(f^{*}) \leq \D(f)$.
    Moreover, by construction, $\gap(f) = \gap(f^*)$.
    
    We will show that  $n-\gap(f^{*}) +1 \leq \D(f^*)$, establishing the lower bound.
    Suppose, for contradiction, that there exists a deterministic algorithm that makes at most $t \le n-\gap(f)$ queries.

Let $x,y\in \dom(f)$ be such that $f(x)\neq f(y)$ and 
\[
|\wt(x)-\wt(y)|=\gap(f).
\]
Without loss of generality, assume $\wt(x)\le \wt(y)$.

We construct an adversarial transcript as follows.
The first $\wt(x)$ distinct queries output $1$ and afterwards, the output is $0$.
Let $S\subseteq [n]$ be the set of queried positions, with $t=|S|\leq n-\gap(f)$.

We now construct two inputs $x_0$ and $x_1$ that are both consistent with this transcript but satisfy $f(x_0)\neq f(x_1)$.

Let $t_1=\min\{t,\wt(x)\}$ be the number of answers equal to $1$.

For the unqueried positions $[n]\setminus S$, whose size is $n-t \ge \gap(f)$, we define:
\begin{itemize}
    \item $x_0$ by placing exactly $\wt(x)-t_1$ ones among the unqueried positions (and zeros elsewhere),
    \item $x_1$ by placing exactly $\wt(y)-t_1$ ones among the unqueried positions.
\end{itemize}

Notice that both constructions are feasible.

By construction, both $x_0$ and $x_1$ agree with the transcript on $S$, so the algorithm cannot distinguish them.
Moreover,
\[
\wt(x_0)=\wt(x), \qquad \wt(x_1)=\wt(y),
\]
and hence, by symmetry of $f$, we have
\[
f(x_0)=f(x) \neq f(y)=f(x_1).
\]

Thus, the algorithm produces the same output on two inputs with different function values, contradicting correctness.
Therefore, $\D(f)\ge n-\gap(f)+1$.

For the matching upper bound, we give the trivial algorithm: query any set of $t = n-\gap(f)+1$ positions.
Suppose the transcript contains $t_1$ ones and $t - t_1$ zeros.

Any completion of the remaining $n-t = \gap(f)-1$ unqueried positions can contribute at most $\gap(f)-1$ additional ones.
In particular, any two inputs consistent with the same transcript differ in Hamming weight by at most $\gap(f)-1$.
By definition of $\gap(f)$, no two such inputs can be conflicting.
Hence, any transcript of size $t$ uniquely determines $f(x)$.
Therefore, $\D(f) \le n-\gap(f)+1$, completing the proof.
\end{proof}

Combining Lemma~\ref{lem:Q_for_sym} and Lemma~\ref{lem:d_for_sym} we can completely characterize $\D$ vs $\Q$ behavior partial symmetric functions.

\subsection{Functions defined on a slice}\label{sec:balanced_slice}

In this section, we will consider partial functions $f:\slice{n,k}\to\{0,1\}$ whose domain is a single slice of the Boolean cube.
The deterministic and quantum query complexities of such functions were already considered by~\cite{bendavid2014structurepromisesquantumspeedups}, who showed that $\D(f)=O(\Q(f)^{18})$.
Here, we prove an alternate bound on $\D(f)$.
Our main motivation is to give a fine-grained bound based on the slice's location.
As such, our result is not comparable to that of  \cite{bendavid2014structurepromisesquantumspeedups}.
Their result is tighter for slices near the edges, whereas our bound is tighter for slices near the middle section.
In particular, we give a tighter bound for roughly balanced domains. 
\begin{definition}[Balanced blocks] 
We say that a block $B \subseteq [n]$ is balanced for $x \in \lbrace 0,1 \rbrace^n$ if \[
       \ \abs{\{i\in [n]:i\in B\text{ and } x_i=1\}} = \abs{\{i\in [n]: i\in B\text{ and } x_i=0\}}.\]
\end{definition}

\begin{remark}\label{rem:balanced}
    For a function $f:\slice{n,k}\to\{0,1\}$ and an input $x\in\slice{n,k}$, we have $x^B\in\slice{n,k}$ if and only if $B$ is balanced for $x$. Hence, the block sensitivity $\balancedblocksensitivity(f)$ (Definition \ref{def:sensitivity_perp}) can be restricted to only balanced blocks.
\end{remark} 

\begin{lemma}\label{lemma:slice} Let $f:\slice{n,k}\to\{0,1\}$ for some fixed $k\in [n-1]$.
Then,
    $\D(f) \leq \frac{3n}{\min\{k, n-k\}} \certificate{f} \balancedblocksensitivity(f)$.
\end{lemma}
\begin{proof}
    We follow a similar structure to the proof of $\D(f)\leq \certificate{f} \balancedblocksensitivity(f)$ for total functions~\cite{beals2001quantum}. The main difference is that in order to turn the indices which disagreed with the queried certificate into blocks, by~\Cref{rem:balanced}, we must ensure they are balanced. We \say{balance} them out by adding indices which were not queried by the algorithm.
    
    We may assume $f$ is non-constant and that $k\leq n/2$ as the case for $k> n/2$ follows by symmetry. Furthermore, one may observe that for all functions $f$ on $\slice{n,k}$, we have that $\certificate{f} \leq k$. This is due to the fact that for any $x\in\slice{n,k}$, every other input which agrees with the $k$ indices $i$ such that $x_i=1$ is either outside of $\slice{n,k}$ or $x$ itself. The deterministic algorithm proceeds as follows.
    
    \begin{algorithm}\caption{}
        Let $q$ denote the partial assignment which records all the bits queried so far. Start with the trivial assignment, $q=*^n$ and run the following subroutine $\balancedblocksensitivity(f)$-times.
        \begin{enumerate}
            \item Let $r\in[\balancedblocksensitivity(f)]$ be the current round and $c_r$ be a $\{1,*\}$-certificate of size at most $\bperpcertificate{1}{f}$ which agrees with $q$. If no such $c_r$ exists, output $0$.
            \item Query $c_r$ and update $q$ accordingly. If our queries agree with $c_r$, output $1$.
            \item If $\abs{\mathrm{Support}(q)} \geq \frac{k}{3}$, query the remainder of 
            $x$ and output $f(x)$.
        \end{enumerate}
    
        If we have not output a value after $\balancedblocksensitivity(f)$ rounds, all inputs in $\slice{n,k}$ which agree with $q$ must map to the same value $b$. Therefore, output $b$.
    \end{algorithm}

    Let us prove the correctness of the algorithm. If we terminate during the loop, the correctness follows by construction. Hence, it suffices to show that after $\balancedblocksensitivity(f)$ rounds, all inputs that agree with $q$ have the same output. Thus, we may hard-code the output for this sub-case. Assume the statement is false, meaning there is some $y\in\slice{n,k}$ agreeing with $q$ such that $f(y)\neq f(x)$. We show that this implies that there are more than $\balancedblocksensitivity(f)$ balanced blocks for some input $y$. Let $U=\mathrm{Support}(q)$ and $u=\abs{U}<k/3$.

    Without loss of generality, let $f(y)=0$. Firstly, as $f(x)\neq f(y)$ and $x,y\in\slice{n,k}$, there must exist a balanced block $B^\prime\subseteq[n]\setminus U$ such that $x^{B^\prime}=y$. Suppose we turn $x$ into $y$ by successively flipping pairs of variables $b$ in $B^\prime$ such that $\abs{b}=2$ and $b$ is balanced. By definition of $B^\prime$, letting $\Tilde{x}$ being the current input, we must come across some $b$ such that, for which $f(\Tilde{x})\neq f(\Tilde{x}^{b})$. As both $\Tilde{x}, \Tilde{x}^b\in\slice{n,k}$, $b$ is a balanced sensitive block for both inputs. Without loss of generality, we may assume that $x=\Tilde{x},y=\Tilde{x}^b$, with $B^\prime=b$ and hence $\abs{B^\prime}=2$.
    
    Let $i_{r}$ be the set of indices on which $q$ disagreed with $c_r$ and $z_{r,b} = \{j\in i_r: y_j=b\}$. As the algorithm did not terminate during the loop, $i_r\neq \emptyset$. As each $c_r$ agrees with the preceding transcript, the sets $i_r\subseteq U$ are pairwise disjoint. Since $\certificate{f}\leq k$, we have that $\abs{i_r} \leq k$. Per~\Cref{rem:balanced}, we shall argue that we may create $\balancedblocksensitivity(f)$ balanced blocks $B_r$ using $i_r$. Specifically, iterating $r$ from $1$ to $\balancedblocksensitivity(f)$, let us describe some balancing set $s_r$ which we use to define the block $B_r = i_r \cup s_r$ as follows:
    \begin{itemize}
        \item If $\abs{z_{r,1}} > \abs{z_{r,0}}$, let $s_r$ be some subset of $\{j\in [n]: y_j = 0 \text{ and } j\not\in (\cup_{l\in [r-1]} B_l )\cup U\cup B^\prime\}$ such that $\abs{s_r} = \abs{z_{r,1}} - \abs{z_{r,0}}$.
        \item If $\abs{z_{r,1}} = \abs{z_{r,0}}$, let $s_r=\emptyset$.
        \item If $\abs{z_{r,1}} < \abs{z_{r,0}}$, let $s_r$ be some subset of $\{j\in [n]: y_j = 1 \text{ and } j\not\in (\cup_{l\in [r-1]} B_l )\cup U\cup B^\prime\}$ such that $\abs{s_r} = \abs{z_{r,0}} - \abs{z_{r,1}}$.
    \end{itemize}
    The total number of balancing indices required is at most $\sum_r|i_r|\leq u$. Since $y$ has at least $k$ coordinates of either value, at least $k-2u-1\geq u$ remain eligible at each step, using $k\geq3u+1$. In each case, we find that the set $B = \{B_1,\dots B_{\balancedblocksensitivity(f)}, B^\prime \}$ satisfies the properties of pairwise disjoint balanced sensitive blocks at $y$. Lastly, if we don't terminate by reaching the $k/3$ threshold of queried values, we query at most $\bperpcertificate{1}{f} \balancedblocksensitivity(f) \leq \certificate{f} \balancedblocksensitivity(f)$ indices. On the other hand, if we reach the threshold, we find that $\bperpcertificate{1}{f} \balancedblocksensitivity(f) \geq k/3$, completing the proof.
\end{proof}

Therefore, by combining \Cref{lem:degree_lower_bound,lem:bs_leq_approxdeg,lem:slice_c_bs},
\begin{align*}
    \D(f) = O\left(\frac{n}{\min\{k, n-k\}} \Q(f)^6\right).
\end{align*}

%% file: extension.tex
\section{Quantum speedups through completions}\label{sec:quantum_speedup_completions}

\subsection{Completion complexity}
\label{sec:completion}

As we have mentioned before, we wish to find ways to show that a given function $f$ either admits a superpolynomial separation between $\Q(f)$, $\R(f)$ and $\D(f)$ or has a polynomial relationship between these measures. In this section, we show results in this direction by giving a new technique for showing non-speedup results for any partial function $f$ by exhibiting a completion $F$ (which we define below) that satisfies certain properties.

For a partial function $f:\domain{f}\rightarrow \{0,1\}$ we say that the total function $F: \{0,1\}^n\rightarrow \{0,1\}$ is a \emph{completion} of $f$ if $F(x)=f(x)$ for all $x \in \domain(f)$.
We denote the set of all completions of $f$ as $\widehat{C}(f)$. 

\begin{definition}[Completion complexity of $f$ with respect to query complexity measure $M$]
Let $f$ be a partial Boolean function and $M$ be a query complexity measure.
Then, the completion complexity of $f$ (with respect to $M$), denoted as $\cmplcom{M}(f)$ is 
\[ \cmplcom{M}(f) = \min_{F \in \widehat{C}(f)} M(F).  \]
\end{definition}
  For a partial function $f$ we use $F_M(f)$ to denote the completion that achieves the above minimum\footnote{If there are multiple completions that achieve the minimum, choose any one of them.}. 
When $f$ is clear from context, we drop it and simply write $F_M$.
We say that $F_M$ is an optimal completion of $f$ with respect to $M$.

\begin{definition}[Polynomially tight measures]
We say that measures $M_1,M_2, \ldots, M_l$ are polynomially tight if for all total functions $F$, we have, $M_1 (F)  \polyeq M_2(F) \polyeq \cdots \polyeq M_{l} (F)$.
\end{definition}

\begin{observation} \label{lem:completions of polynomially tight measures}
    Let $M_1,M_2$ be two measures that are polynomially tight, then for any 
    (possibly partial) function $f$, we have 
    $\cmplcom{M_1}(f) \polyeq \cmplcom{M_2}(f)$.
\end{observation}
\begin{proof} It suffices to show that $\cmplcom{M_1}(f) \polyleq \cmplcom{M_2}$.
The other inequality $\cmplcom{M_2}(f) \polyleq \cmplcom{M_1}$, follows similarly.
Recall that $F_{M_1}$ and $F_{M_2}$ are completions that achieve $\cmplcom{M_1}(f)$ and $\cmplcom{M_2}(f)$, respectively. Therefore,
    \begin{align*}
        \cmplcom{M_1}(f) = M_1(F_{M_1}) &  \leq M_1(F_{M_2})\\
        & \polyeq M_2(F_{M_2})\\
        & = \cmplcom{M_2}(f).
    \end{align*}
The first inequality comes from the definition of $\cmplcom{M_1}(f)$ and the second equality is because $M_1$ and $M_2$ are polynomially tight. 
\end{proof}
The statement above shows that for polynomially tight measures, their optimal completions are polynomially tight. In fact, we can prove a stronger statement. The following lemma shows that the quantities remain polynomially tight when measures are considered not only with their optimal completions, but also with respect to optimal completions of any other measure.

\begin{lemma}
\label{lem:polymeasures}
    Let $\mathcal{M} = \lbrace M_1, M_2, \ldots, M_ {\vert\mathcal{M} \vert} \rbrace$ be a set of polynomially tight measures.
    Let $f$ be a partial function and for $i \in [\abs{\mathcal{M}}]$, let \[\mathcal{F} = \lbrace F_{i} \ : \  F_{i}  \text{ is the optimal completion of } f \text{ with respect to } M_i \in \mathcal{M} \rbrace.\]
    Then, for every $M_i, M_j \in \mathcal{M}$ and every $F_k, F_l \in \mathcal{F}$, we have, $M_i(F_k) \polyeq M_j(F_l)$.    
    \end{lemma}
\begin{proof}
For ease of readability, we show this with $(i,j,k,l)=(1,2,3,4)$, the proof remains same for any values of $i,j,k,l$.
Again, we show  $M_1(F_{M_3}) \polyleq M_2(F_{M_4})$, $M_2(F_{M_4}) \polyleq M_1(F_{M_3})$ follows analogously.
\begin{align*}
    M_1(F_{M_3})  &  \polyeq  M_3(F_{M_3}) \\
    & \leq M_{3}(F_{M_4}) \\
    & \polyeq M_4(F_{M_4}) \\ 
    & \polyleq M_2(F_{M_4}). \qedhere
\end{align*}
\end{proof}
Consider any partial function $f$ and let $\mathcal{A}$  be the optimal deterministic algorithm.
Note that $\mathcal{A}$ produced some fixed outcome for each $x \notin \dom(f)$ as well. Thus we can define a completion of $f$, defined via $F_D(x)= \mathcal{A}(x)$. Thus, we get the following:
\begin{observation}\label{fact:dequalscompletion}
    For any partial function $f$, $D(f) = D(F_D)$.
\end{observation}

Combining everything, we get the following lemma: a measure $M$ is extendable without a blow-up for $f$ if and only if it is polynomially related to deterministic query complexity.

\begin{lemma}\label{lemma:general}
    For any \emph{partial function} $f:\Dom(f) \to \{0,1\} $, and for any complexity measure $M$ which is polynomially related to $\D$ for \emph{total functions},
    \[ \D(f) \polyeq M(f) \textit{ iff 
 } \cmplcom{M}(f) \polyeq M(f).  \]
\end{lemma}

In particular, this lemma gives us a way to establish superpolynomial speedup between $\D$ and \emph{an arbitrary measure} $M$ by showing that the measure $M$ is not completable to a total function without a superpolynomial blow-up.

 \begin{proof}
     Assume $D(f) \polyeq M(f)$, the we get the following chain of equalities.
     \[ \cmplcom{M}(f) = M(F_M) \polyeq \D(F_D) = D(f) \polyeq M(f). \]
     The first equality follows from the definition of $F_M$, the second one from \Cref{lem:polymeasures}, the third one from \Cref{fact:dequalscompletion} and the last one follows from the assumption.

     The other direction (assuming $M(f) \polyeq \cmplcom{M}(f)$) also follows via a similar chain of equalities; we write them here for the sake of completion.
     \begin{align*}
         M(f) \polyeq \cmplcom{M}(f) = M(F_M) \polyeq \D(F_D) = D(f).&\qedhere
     \end{align*}
 \end{proof}
 
\subsection{Towards characterizing the \emph{na\"ive} completion}

When considering methods for completing a partial function $f: \domain{f} \rightarrow \{0,1\}$ to a total function, the first thought one might have is to set the output value of every point $x\notin \domain{f}$ to a single value, either $f(x)=0$ or $f(x)=1$ for all $x\notin \domain{f}$. We call this the \emph{na\"ive} completion. We help characterize when the na\"ive completion is possible by showing that one can do this if there exists an efficient approximation of the indicator polynomial $\mathbbm{1}_{\domain{f}}(x)$ which is $1$ if $x\in \domain{f}$ and $0$ otherwise.

\begin{lemma}\label{lemma:naive_completion}
Let $f: \domain{f} \rightarrow \{0,1\}$ be a partial Boolean function, and let $p(x)$ be a degree-$d$ polynomial that $\epsilon$-approximates $f$ on $\domain{f}$, for some constant $\epsilon \in (0,1/3)$.  
Define the na\"ive completion
\[
F_0(x) = \begin{cases} 
f(x) & x \in \domain{f} \\
0 & x \notin \domain{f}.
\end{cases}
\]
Suppose there exists an algorithm that decides whether $x \in \domain{f}$ using at most $q$ queries to the input bits. Then the approximate degree of $F_0$ satisfies $\approxdegree(F_0) \leq \mathrm{poly}(d,q)$.
\end{lemma}

\begin{proof}
Since $p(x)$ $\epsilon$-approximates $f$ on $\domain{f}$, we have:
\[
|p(x) - f(x)| \leq \epsilon \quad \text{for all } x \in \domain{f}.
\]
Let $g(x) = \mathbbm{1}_{\domain{f}}(x)$ be the indicator function on $\domain{f}$. By a standard result, if $g(x)$ can be decided by a $q$-query algorithm, then there exists a polynomial $q^\prime(x)$ that $\epsilon'$-approximates $g(x)$, for some constant $\epsilon' \in (0,1/3-\epsilon)$, such that $\deg(q^\prime) = O(q)$~\cite{beals2001quantum}. That is,
\[
\begin{cases}
|q^\prime(x) - 1| \leq \epsilon' & \text{if } x \in \domain{f}, \\
|q^\prime(x) - 0| \leq \epsilon' & \text{if } x \notin \domain{f}.
\end{cases}
\]

Now define the polynomial:
\[
r(x) \coloneqq p(x) \cdot q^\prime(x).
\]
We claim that $r(x)$ is a $(\epsilon + \epsilon')$-approximating polynomial for $F_0(x)$.

\textbf{Case 1:} $x \in \domain{f}$ \\
Then $F_0(x) = f(x)$ and $|p(x) - f(x)| \leq \epsilon$. Also, $|q^\prime(x) - 1| \leq \epsilon'$. So,
\begin{align*}
|r(x) - f(x)| &= |p(x)q^\prime(x) - f(x)| \\
&= |p(x)q^\prime(x) - p(x) + p(x) - f(x)| \\
&\leq |p(x)| \cdot |q^\prime(x) - 1| + |p(x) - f(x)| \\
&\leq \epsilon' + \epsilon.
\end{align*}

\textbf{Case 2:} $x \notin \domain{f}$ \\
Then $F_0(x) = 0$, and we have:
\[
|r(x) - F_0(x)| = |p(x)q^\prime(x)| \leq |q^\prime(x)| \cdot |p(x)| \leq \epsilon'.
\]

Thus, in both cases, $r(x)$ approximates $F_0(x)$ within additive error at most $\epsilon + \epsilon'$. Finally, since $\deg(q^\prime) = O(q)$, the degree of $r(x)$ is:
\[
\deg(r) \leq \deg(p) + \deg(q^\prime) \leq \poly(d, q).
\]

Therefore, $\approxdegree(F_0) \leq \mathrm{poly}(d,q)$. 
\end{proof}

We note that an analogous statement may be shown for the other na\"ive completion $F_1(x)$ by repeating the same proof with $r(x)=p(x) q^\prime(x)+1-q^\prime(x)$.

\begin{corollary}
    Let $f$ be any partial function with domain $\domain{f} \subsetneq \lbrace 0,1 \rbrace^n$.
    If there exists an algorithm\footnote{it does not matter whether the algorithm in deterministic, randomized or quantum as all are polynomially related.} $\mathcal{A}$ for deciding $x \in \domain{f}$ with query complexity $q \polyleq \approxdegree(f)$, then $f$ has no quantum speedup.
\end{corollary}

The question then is when can we efficiently distinguish between inputs inside and outside of $\domain{f}$ in query complexity that is $\poly(\Q(f))$.

\begin{example}
     For any $T\leq n/2$, let $g: \{0,1\}^{T}\rightarrow \{0,1\}$ be an arbitrary function with $\approxdegree(g)=T$. Let $h: \domain{h} \rightarrow \{0,1\}$ be a nonconstant partial function with $\domain{h}\subseteq \{0,1\}^{n}$ such that whether $x\in h^{-1}(1)$ depends only on the last $T$ locations. We let $f: \domain{f} \rightarrow \{0,1\}$ where $\domain{f}=\{x\in \{0,1\}^{n}: h(x)=1\}$ and $f(x)=g(x_{1},\dotsc,x_{T})$. It is obvious that $\approxdegree(f)=\approxdegree(g)=T$. So, if we define our completion $F(x)=g(x_{1},\dotsc,x_{T})$ if $h(x)=1$ and $F(x)=0$ otherwise, then we also get that $\approxdegree(F)=T$.
\end{example}

\subsection{Completions of slices and unions of slices}\label{sec:slices}

In the following section, we give a completion of an polynomial characterizing a function on the slice, which we use to strengthen the bound of~\cite{bendavid2014structurepromisesquantumspeedups}. Afterwards, we extend this to the union of slices by randomly selecting a single-slice completion. Let us first prove several useful statements.

\begin{proposition}\label{lem:slice_polynomial_packing}
    Fix some real polynomial $p$ and $x\in \{0,1\}^n$ such that $p(x)\in [0,1]$ for all $x$. Suppose that $A_1,\dots, A_m \subseteq [n]$ are pairwise disjoint and for all $i\in [m]$, $p(x^{A_i}) - p(x)$ has the same sign and magnitude at least $\gamma > 0$. Then,
    \begin{align*}
        m \leq \frac{\deg(p)^2}{\gamma}.
    \end{align*}
\end{proposition}
\begin{proof}
    For $z\in \{0,1\}^m$, let $q(z) = p(x^{\cup_{i: z_i = 1} A_i})$ which can be multilinearized by replacing each variable in $[n]$ with a constant, $z_i$ or $1-z_i$. Notice that $q$ outputs values in $[0,1]$ for all $\{0,1\}^m$. For $s\in [0,1]$, let $H(s) = \E[q(s_1, \dots, s_m)] = q(s,\dots,s)$ where $s_i$ are independent Bernoulli random variables with threshold $s$. Notice that $\deg(H) \leq \deg(p)$ and,
    \begin{align*}
        H^\prime(0) = \sum_{i\in [m]} (p(x^{A_i}) - p(x)).
    \end{align*}
    Furthermore, $H(s)\in [0,1]$. Therefore by Markov's inequality, $m\gamma \leq \abs{H^\prime(0)} \leq \deg(p)^2$.
\end{proof}

\begin{lemma}\label{lem:slice_edge_completion}
    For some $t\in [n]$ and partial function $f$ such that $\domain(f) \subseteq \{x\in \{0,1\}^n: \wt(x)\leq t\}$. The there exists a completion $F$ of $f$ such that $\D(F) \leq 3\approxdegree(f)^2 \cdot (t+1)^2$.
\end{lemma}
\begin{proof}
    Fix some approximating polynomial $p$ of $f$ which attains $\deg(p)=\approxdegree(f)$. For some partial assignment $\rho$, let $U\subseteq [n]$ be the set of free coordinates and $u=t-\abs{ \{i\in \rho: \rho(i) = 1 \}}$. Notice that the restricted polynomial $p|_\rho$ has degree at most $\approxdegree(f)$ and approximates $f|_\rho$ on the promise.

    Let $b=0$ if $p|_\rho(0) \leq 1/2$ and $b=1$ otherwise. We say that $A\subseteq U$ flips the value when $f|_\rho(0^A) \neq b$. Notice that for any such $A$, $\abs{p|_\rho(0^A) - p|_\rho(0)} \geq 1/6$ and always have the same sign. By \Cref{lem:slice_polynomial_packing}, there is at most $6\approxdegree(f)^2$ such subsets which are pairwise disjoint. Choose the maximal such set and let $J$ be their union. Then,
    \begin{align}
        \abs{J} \leq 6 u \cdot \approxdegree(f)^2 \label{eq:bound_added_j}
    \end{align}

    Query all coordinates in $J$ and let $r$ be the number of $1$ received. If $r=0$,, then output $b$ as there are no sets which flip. Otherwise, one must exist, so restrict $\rho$ by the answers to $J$ and repeat. Notice that this must strictly decrease $u$. If $r>u$, then output an arbitrary output. Otherwise, if $r=u$, there may exist only one promise value, so output its value. Repeat this process until $J=\emptyset$. Notice that by \Cref{eq:bound_added_j}, the number of queries made by this deterministic procedure is at most,
    \begin{align*}
        6\approxdegree(f)^2 \sum_{u\in [t]} u &= 3\approxdegree(f)^2 \cdot t(t+1).\qedhere
    \end{align*}
\end{proof}
These statements allow us to prove the completion result.

\begin{theorem}\label{thm:single_slice}
    For any partial function $f:\slice{n,k}$, there exists a completion $F$ of $f$ such that,
    \begin{align*}
        \deg(F)= O\left(\approxdegree(f)^2 \bs(f)^4\right).
    \end{align*}
    In particular, $\D(f)=O(Q(f)^{10})$.
\end{theorem}
\begin{proof}
    Assume the function is not constant as otherwise the result trivially follows. Similarly, assume that $k\leq n/2$ as the other case follows by symmetry. If $k\geq 4\bs(f)^2$, then \Cref{thm:slice_bs_d_bound,fact:dequalscompletion} we obtain the statement. Otherwise, \Cref{lem:slice_edge_completion} provides us with a completion $F$ of $f$ such that,
    \begin{align*}
        \D(F) &\leq 3\approxdegree(f)^2 (k+1)^2\\
        &\leq 48 \approxdegree(f)^2 \bs(f)^4.\qedhere
    \end{align*}
\end{proof}

Next, we show that the completion of a single slice behaves similarly on nearby slices.

\begin{lemma}\label{lem:slice_path_transport}
    Consider a partial function $f:\unionSlice{K}\to \{0,1\}$, $k,l\in K$ such that $k\neq l$ and $x\in \slice{n,k}$. Let $T_l$ be some algorithm computing $f|_{\slice{n,l}}$ with at most $t_l$ queries. If,
    \begin{align*}
        t_l + (\bs(f) + 1)\abs{l-k} \leq \begin{cases}
            n-k &\text{if $l > k$}\\
            k &\text{otherwise,}
        \end{cases}
    \end{align*}
    then $T_l(x) = f(x)$.
\end{lemma}
\begin{proof}
    Suppose that $l > k$ and let $J$ be the queried inputs by $T_l$ when given $x$. Notice that the number of $0$ coordinates outside of $J$ is at least $n-k - \abs{J} \geq n-k - t_l \geq (\bs(f) + 1)(l-k)$. Choose pairwise disjoint sets $A_1,\dots,A_{\bs(f)+1}$ of these coordinates each of which with size $l-k$. Notice that $x^{A_i}\in \slice{n,l} $ and agrees with $x$ on $J$, so its output is the same as $T_l(x)$. Notice that $T_l(x)=f(x)$ as $T_l(x)\neq f(x)$ would imply that $x$ has $\bs(f)+1$ disjoint sensitive blocks at $x$, a contradiction. The case when $l < k$ follows analogously by letting $J$ be the set of unqueried $1$ coordinates and $A_i$ have size $k-l$ instead.
\end{proof}

Finally, we obtain the result of the section for unions of slices.

\begin{theorem}\label{thm:multislice}
    For any partial function $f:\unionSlice{K}\to \{0,1\}$,
    \begin{align*}
        \R(f) = O\left( \bs(f)^2 + \max_{k\in K} \D(f|_{\slice{n,k}})\right).
    \end{align*}
    Consequently, $\R(f) = O(Q(f)^{10})$.
\end{theorem}
\begin{proof}
    Assume the function is not constant as otherwise the result trivially follows. For every $l\in K$, let $T_l$ be the optimal algorithm on $f|_{\slice{n,l}}$. We may assume $T_l(x)\in \{0,1\}$ for all $x\in \{0,1\}^n$ due to \Cref{fact:dequalscompletion}. Assume that for all $l\in K$, $n\geq 4t_l$ as otherwise the statement holds by quering every coordinate. Let $\nu$ be,
    \begin{align*}
        \nu &\coloneq \frac{1}{16(\bs(f)+ 1)} \leq \frac{1}{32}.
    \end{align*}
    Consider taking $192(\bs(f)+1)^2$ independent samples of coordinates with replacement and let $u^\prime$ be the fraction of $1$ sampled. Notice that $\mathrm{Var}(u^\prime) \leq \tfrac{\nu^2}{3}$, so by Chebyshev's inequality,
    \begin{align}
        \Pr[\abs{u^\prime - \frac{\wt(x)}{n}} \leq \nu] \geq \frac{2}{3}.\label{eq:slice_weight_estimate_event}
    \end{align}
    Let the weight $l\in K$ be,
    \begin{align*}
        l = \begin{cases}
            \max{j\in K: j/n \leq u^\prime + \nu} &\text{if $u^\prime \leq \tfrac{1}{2}$,}\\
            \min{j\in K: j/n \geq u^\prime - \nu} &\text{otherwise.}
        \end{cases}
    \end{align*}
    If not such $j$ in the equation exists, output an arbitrary bit. Otherwise, run the algorithm $T_l$ on $x$. Suppose that \Cref{eq:slice_weight_estimate_event} occurs and let $k=\wt(x)$ be the actual hamming weight. If $l=k$, then $T_l(x)$ must output the correct value. Otherwise, even if $l\neq k$, notice that $l$ and $k$ satisfy the conditions in \Cref{lem:slice_path_transport}, so even in this case $T_l(x) = f(x)$. Hence the algorithm is correct conditioned on the event in \Cref{eq:slice_weight_estimate_event}, meaning its success probability is at least $\tfrac{2}{3}$. Finally, the relationship between $\R(f)$ and $\Q(f)$ follows from \Cref{thm:single_slice}.
\end{proof}

\subsection{Towards characterizing the \emph{natural} completion}\label{sec:smoothness}

A natural path when searching for completions of a partial function $f:\domain{f}\rightarrow \{-1,1\}$ is via an approximating polynomial $p(x)$. This motivation is two-fold. Firstly, if $p(x)$ for some $x\not\in \domain{f}$ is sufficiently close to $1$ or $-1$, then via standard boosting techniques we could add $x$ to the domain of $f$ without a super-polynomial increase in the degree of $p$. Secondly, fixing a polynomial allows us to use tools from Boolean analysis. Specifically, the tools rely on using a fixed function, since for a partial function function $f(x)$, the standard tools do not apply due to the fact that $f(x)$ does not possess a unique Fourier expansion, although importantly, a fixed polynomial $p(x)$ does. Given this, we will use $p(x)$ as a \emph{natural} way to characterize the structure of a partial function, and use it to formally define a completion, which we denote as the \emph{natural completion}.

\begin{definition}[Natural completion]
    Consider $f:\domain{f}\rightarrow\{-1,1\}$ with approximating polynomial $p(x)$. The natural completion $F(x)$ of $f(x)$ is,
    \begin{align*}
        F(x) = \begin{cases}
            1 &p(x) \geq 0\\
            -1 &\text{otherwise.}
        \end{cases}
    \end{align*}
    When $\approxdegree(F) = \poly(\approxdegree(f))$, we say that $f(x)$ \emph{admits} the natural completion without a degree blow-up.
\end{definition}

Next, we show that under certain Lipschitz properties, a partial function admits the natural completion without a degree blow-up.

\begin{lemma}
\label{lem:lipschitz}
   Consider $f:\domain{f} \rightarrow \{-1,1\}$ with approximating polynomial $p(x)$ such that $\deg(p) = d$. Assume that for each $x\not\in\domain{f}$, there exists some $y\in\domain{f}$ such that,
   \begin{enumerate}
       \item the distance between $x$ and $y$ is at most $r$ (i.e. $\wt(x-y) \leq r$) and,
       \item there exists some constant $c>0$ such that between $x$ and $y$, $p(x)$ has a Lipschitz constant $L\leq \frac{2}{3r} - \frac{1}{d^{c}r}$.
   \end{enumerate} 
   Then $\abs{p(x)}\geq\frac{1}{d^{c}}$ for all $x\in\{-1,1\}^n$, and $f$ admits the natural completion $F$ with $\approxdegree(F)=O(d^{c+1})$.
\end{lemma}

\begin{proof}
    Let $x\not\in\domain{f}$ and $y\in\domain{f}$ be a pair of inputs which satisfy the conditions of the statement. Let $L$ be the constant such that $x$ and $y$ are $L$-Lipschitz. By definition, we have that,
    \begin{align*}
        \abs{p(x) - p(y)} \leq Lr.
    \end{align*}
    Next, we lower bound $\abs{p(x)}$. By the reverse triangle inequality,
    \begin{align*}
        \abs{p(x)} &\geq \abs{ \abs{p(y)} - \abs{p(x)-p(y)}}\\
        &\geq \abs*{\frac{2}{3} - Lr}.
    \end{align*}
    Let us show that $\abs{p(x)}\geq \frac{1}{d^c}$ for all $x\in\{-1,1\}^n$. This holds for $x\in\domain{f}$ by definition, so assume $x\not\in\domain{f}$.
    To satisfy the condition, we want $\frac{2}{3} - Lr \geq \frac{1}{d^c}$ for some constant $c$. Rearranging for $L$,
    \begin{align*}
        L \leq \frac{2}{3r} - \frac{1}{d^c r}.
    \end{align*}
    Next, we show $\approxdegree(F) = O(d^{c+1})$. Let $P: \mathbb{R}\rightarrow \mathbb{R}$ be the polynomial in~\Cref{lem:chebyshev_sign} for $\delta = \frac{1}{d^c}$ and $\epsilon = \frac{1}{3}$ which has $\deg(P)=O(d^c)$. Consider the polynomial $\widetilde{p}(x) = P(p(x))$. As $\abs{p(x)}\geq \frac{1}{d^c}$, $\widetilde{p}(x)$ approximates $F(x)$ with $\epsilon=\frac{1}{3}$. Secondly, $\deg(\widetilde{p}) = \deg(P) \cdot \deg(p) = O(d^{c+1})$, completing the proof.
\end{proof}

The remainder of the section is concerned with applying~\Cref{lem:lipschitz}. To do so, we will use the following measure which subsumes the variable $r$ in the statement of~\Cref{lem:lipschitz}.

\begin{definition}[Covering radius]
Let $\mathcal{X} \subseteq \lbrace -1,1 \rbrace^n$. The covering radius of $\mathcal{X}$ is defined as follows:
\[ \rc(\mathcal{X}): = \max_{y \in \lbrace -1,1 \rbrace^n } \min_{ x \in \mathcal{X}} d_H(x,y).\]
Informally, if we draw $\rc$-sized Hamming balls around the set $\mathcal{X}$, then it covers the entire space $\lbrace -1,1 \rbrace^n$.    
\end{definition}

We note that we use $\mathcal{P}_{n,d}(f)$ to denote the set of polynomials in $\mathcal{P}_{n,d}$ that approximate $f$ up to error $\epsilon=\frac{1}{3}$. Furthermore, for simplicity of notation, we will use the following variable $\eta(f,c)$ which is a simplification of the Lipschitz condition in~\Cref{lem:lipschitz},
\begin{align*}
    \eta(f,c) \coloneqq \frac{\frac{2}{3} - \approxdegree(f)^{-c}}{\rc(\domain{f})}.
\end{align*}

From \Cref{lem:lipschitz} and \Cref{def:discrete_derivative}, we immediately get the following corollary.

\begin{corollary}
    Consider $f:\domain{f}\rightarrow \{-1,1\}$ with an approximating polynomial $p(x)$ of degree $d$ such that for all $i\in[n]$, $\norm{D_{i}p}_{\infty}\leq \frac{\eta(f,c)}{2}$ for some constant $c\geq 0$. Then $f$ admits the natural completion $F$ with $\widetilde{deg}(F)=O(d^{c+1})$.
\end{corollary}

We show that functions with an approximating polynomial whose indices have low influence and sparsity do not attain superpolynomial speedup.

\begin{lemma}\label{lem:natural_influence_sparsity}
    Consider $f: \domain{f}\rightarrow \{-1,1\}$ with $\approxdegree(f)=d$. Let $c\geq 0$ be some constant such that,
    \begin{align}
        \min_{p\in \mathcal{P}_{n,d}(f)}\max_{i\in [n]}\, \sqrt{\spar_i(p) \Inf_i(p)}\leq \frac{\eta(f,c)}{2}.\label{eq:influence_bound}
    \end{align}
    Then $f$ admits the natural completion $F$ with $\approxdegree(F) = O(d^{c+1})$.
\end{lemma}

\begin{proof} We may write the approximating polynomial $p(x)$ as,
    \begin{align*}
        p(x)=\sum_{\substack{S\subseteq[n]\\ |S|\leq d}}\hat{p}(S)\chi_{S}(x).
    \end{align*}
    Consider some $i\in [n]$ and let $\mathbbm{1}[i\in S]$ be the indicator function for whether $i\in S$. By definition of $\chi_S(x)$, we have that $\chi_S(x^i) = (-1)^{\mathbbm{1}[i\in S]}\chi_S(x)$. Therefore,
    \begin{align*}
        \chi_S(x) - \chi_S(x^i) = \begin{cases}
            0 & i\notin S\\
            2\chi_{S}(x) & i\in S.
        \end{cases}
    \end{align*}
    Therefore,
    \begin{flalign*}
        &&|p(x)-p(x^{i})|&\leq2\sum_{S\ni i}|\hat{p}(S)||\chi_{S}(x)| &\\
        &&& = 2\sum_{\substack{S\ni i \\ \hat{p}(S)\neq 0}}|\hat{p}(S)|\\
        &&&\leq 2 \sqrt{\sum_{\substack{S\ni i \\ \hat{p}(S)\neq 0}} \hat{p}(S)^2} \sqrt{\sum_{\substack{S\ni i \\ \hat{p}(S)\neq 0}} 1} &\text{(Cauchy-Schwartz)}\\
        &&&=2 \sqrt{\Inf_i(p)}\sqrt{\spar_i(p)} &\text{(\Cref{def:influence})}\\
        &&&\leq \eta(f,c).
    \end{flalign*}
    Thus, $p(x)$ is $\eta(f,c)$-Lipschitz between all pairs of points $x,y\in \{-1,1\}^n$. By~\Cref{lem:lipschitz}, $f(x)$ admits the natural completion.
\end{proof}

By combining the lemma above with previous results, we get the following.

\begin{theorem}\label{thm:natural_influence}
    Consider $f:\domain{f}\rightarrow \{-1,1\}$. Suppose there exists some constant $c$ such that~\Cref{eq:influence_bound} is satisfied. Then $\D(f)\polyeq \Q(f)$.
\end{theorem}
\begin{proof}
    From \Cref{lem:natural_influence_sparsity}, we have that there exists a completion $F$ of $f$ such that $\approxdegree(f)\polyeq \approxdegree(F)$. By \Cref{lemma:general}, this implies that $\approxdegree(f)\polyeq \D(f)$. As $\approxdegree(f) \polyleq \Q(f)\polyleq \D(f)$ for all partial functions, we are done.
\end{proof}

Let us consider an instance of Gap majority where the gap is $\frac{2n}{3}$. Per \Cref{lem:d_for_sym,lem:Q_for_sym}, this instance separates $\Q(f)$ and $\D(f)$. By \Cref{thm:natural_influence}, it should not satisfy the conditions of \Cref{lem:natural_influence_sparsity}.

\begin{example}[Gap majority]
    Consider an instance of gap majority $f: \domain{f} \rightarrow \{-1,1\}$ where,
    \begin{align*}
        \domain(f) = \{x: \wt(x) \geq \frac{5n}{6} \text{ or } \wt(x)\leq \frac{n}{6} \}.
    \end{align*}
    Then $p(x) = \frac{1}{n}\sum x_i$ is a $1/3$-approximating polynomial of $f$. For each $i\in [n]$, we have $\spar_i(p)=1$, $\Inf_i(p) = \frac{1}{n^2}$ and $\eta(f,c) < \frac{2}{n}$. As the bound on $\eta(f,c)$ is a strict inequality, $p(x)$ does not satisfy the conditions in~\Cref{eq:influence_bound}.
\end{example}

The tightness shows that the bound on influence in \Cref{eq:influence_bound} cannot be improved. A similar result to \Cref{lem:natural_influence_sparsity} follows in the setting of the $p$-biased hypercube, a generalization of the standard domain.\footnote{When $p=1/2$, the $p$-biased hypercube represents the Boolean hypercube.}

\begin{lemma}\label{lem:biased}
    Consider the function $f^{p}:\domain{f} \rightarrow \{-1,1\}$ with $\approxdegree(f^p) = d$. Let $\beta = \max\left\{ \sqrt{\frac{p}{1-p}}, \sqrt{\frac{1-p}{p}}\right\}$ and $K_{n,d,p} = \sum_{l = 0}^{d-1} \binom{n-1}{l}\beta^{2l}$. Suppose that,
    \begin{align*}
        \min_{p\in \mathcal{P}_{n,d}(f^p)} \max_{i\in [n]}\, \sqrt{\Inf_i[p]} \leq \frac{\sigma\cdot \eta(f,c)}{2\sqrt{K_{n,d,p}}}.
    \end{align*}
    Then there exists a $p$-biased completion $F^p$ of $f^p$ such that $\approxdegree(F^p) = O(d^{c+1})$.
\end{lemma}
\begin{proof}
    By~\Cref{def:p_biased_basis}, we have that for any $S\subseteq [n]$ and $i\in S$, $\phi_S(x^i) = \phi(- x_i)\phi_{S\setminus \{i\}}(x)$. As $p(x) = \sum_{\abs{S}\leq d} \hat{p}(S)\phi_S(x)$,
    \begin{align*}
        \abs{p(x) - p(x^i)} &=\abs*{\sum_{S\ni i} \hat{p}(S)(\phi_S(x) - \phi_S(x^i))}\\
        &=\abs*{\sum_{S\ni i} \hat{p}(S)(\phi(x_i) - \phi(-x_i))\phi_{S\setminus \{i\}}(x)}\\
        &= \abs*{\phi(x_i) - \phi(-x_i)} \cdot \abs*{\sum_{S\ni i} \hat{p}(S) \phi_{S\setminus \{i\}}(x)}\\
        &\leq \frac{2}{\sigma} \left(\sum_{S\ni i}\hat{p}(S)^2 \right)^{1/2} \left(\sum_{\abs{S}\leq d; S\ni i} \phi_{S\setminus \{i\}}(x)^2 \right)^{1/2},
    \end{align*}
    where the last inequality is an application of the Cauchy-Schwartz inequality and the fact that by~\Cref{def:p_biased_phi}, we have that $\abs*{\phi(x_i) - \phi(-x_i)} = \frac{2}{\sigma}$. Next, we bound each term. Notice that for any $i\in [n], \abs{\phi(x_i)} = \abs{\frac{x_i - \mu}{\sigma}} \leq \beta$. Therefore,
    \begin{align*}
        \sum_{\substack{\abs{S} \leq d\\ S\ni i}} \phi_{S\setminus \{i\}}(x)^2 \leq \sum_{l=0}^{d-1} \binom{n-1}{l} \beta^{2l} = K_{n,d,p}.
    \end{align*}
    Plugging everything together and using the definition of influence (\Cref{def:influence}),
    \begin{align*}
        \abs{p(x)-p(x^i)} &\leq \frac{2}{\sigma} \frac{\sigma\cdot\eta(f,c)}{2\sqrt{K_{n,d,p}}} \sqrt{K_{n,d,p}} = \eta(f,c).
    \end{align*}
    Therefore, $p(x)$ is $\eta(f,c)$-Lipschitz. By~\Cref{lem:lipschitz}, there exists a $p$-biased completion $F^p$ of $f$ such that $\approxdegree(F^p) = O(d^{c+1})$.
\end{proof}

\subsection{Towards general partial function completions}

In this section we explore the idea of completions of polynomials for partial Boolean functions.

\subsubsection{Completability of exact degree}\label{sec:exact_deg}

When considering quantum query complexity, one is often concerned with the minimum degree of a polynomial that approximated $f$. However, these polynomials do not have the same uniqueness as exact polynomials representing $f$, and can be slightly more difficult to work with~\cite{odonnell2021analysisbooleanfunctions}. For this reason, before we discuss techniques towards general completion of approximate degree for partial functions, we begin by considering how to complete the degree of an exact polynomial representing a partial Boolean function $f$.

\begin{lemma}
    Consider $f:\domain{f}\rightarrow \{0,1\}$ such that $\deg(f)=d$ and let $F$ be a completion of $f$. Then $\deg(F)=d$ if and only if for every $S\subseteq [n]$ such that $\abs{S} > d$, 
    \begin{align*}
        \abs{\{x: F(x) = 1, \wt(x) \text{ is even and } x\subseteq S\}} = \abs{\{x: F(x) = 1, \wt(x) \text{ is odd and } x\subseteq S\}}.
    \end{align*}
    where $x\subseteq S$ denotes $\{i\in [n]: x_i=1\} \subseteq S$.
\end{lemma}
\begin{proof}
    ($\Rightarrow$) Assume $\deg(F) = d$ and let $p(x) = \sum_{S\subseteq [n]} c_S x^S$ be the polynomial representation of $F$. By the M\"obius inversion formula,
    \begin{align*}
        c_S = \sum_{T\subseteq S} (-1)^{|S|-|T|}F(T) = (-1)^{\abs{S}} \sum_{T\subseteq S} (-1)^{\abs{T}} F(T).
    \end{align*}
    Notice that since $\deg(F) = d$, $c_S = 0$ for all $S$ such that $\abs{S} > d$. Therefore,
    \begin{align*}
        \sum_{T\subseteq S, \abs{T} \text{ is even}} F(T) = \sum_{T\subseteq S, \abs{T} \text{ is odd}} F(T).
    \end{align*}
    ($\Leftarrow$) Assume the condition is satisfied. Then from the M\"obius inversion formula, $c_S=0$ for all sets $S$ such that $\abs{S}>d$, meaning that $\deg(F) =d$.
\end{proof}

\subsubsection{Perturbing approximate polynomials}\label{sec:pertubing_approx_polynomial}

We turn from the case of completing exact degree polynomials to a discussion about working with approximating polynomials. The reason for pivoting to approximating polynomials is the same as in \Cref{sec:smoothness}, namely that an approximating polynomial $p(x)$ provides us with a concrete tool to study the partial function $f:\domain{f}\rightarrow \{-1,1\}$.

The main idea behind the proofs in \Cref{sec:smoothness} is that by ensuring that for all $x\not\in \domain{f}$ $p(x)$ is sufficiently far away from $0$, we may apply standard boosting techniques to turn it into an approximating polynomial of the natural completion of $f$. Difficulty arises when $p(x)$ is extremely close to $0$ for some $x\not\in \domain{f}$. Formally, we mean that $\abs{p(x)} \leq \epsilon$ for some $\epsilon < 1/\poly(\deg(p))$. In this case, boosting as was done with the natural completion would blow-up the degree of the resulting approximating polynomial well beyond $\approxdegree(f)$.

A natural idea to avoid this issue is by taking the approximating polynomial $p(x)$ and adding a small \emph{perturbation} to it which removes all problematic inputs $x$ such that $\abs{p(x)}\leq \epsilon$ without introducing new ones. Specifically, we mean that we perturb the Fourier coefficients of $p(x)$ so that it becomes a valid approximating polynomial for some completion of $f$. For example, let $\lambda,\epsilon$ be some positive constants such that $\lambda+\epsilon\leq \frac{1}{2}$ and $\epsilon \geq \poly(\deg(p))^{-1}$. Suppose that for all $x\in \{-1,1\}^{n}$, $p(x) \not\in [\lambda-\epsilon,\lambda+\epsilon]$. Then we may define the perturbed polynomial $p_\Delta(x) = p(x) - \lambda$, meaning that there are no inputs such that $\abs{p_\Delta(x)} \leq \epsilon$. Therefore, we conclude that $f$ admits a completion without a degree blow-up. Let us begin by defining the following quantity.

\begin{definition}[Domain evaluation matrix]
    Consider $f: \domain{f}\rightarrow \{-1,1\}$ where $m = \abs{\domain{f}}$, $s = \approxspar(f)$ and $p$ is the approximating polynomial which achieves $\approxspar(f)$. 
    The \emph{domain evaluation matrix} $A_{D}(p)$ is the $m \times s$ matrix defined as,
    \begin{align*}
        (A_{D}(p))_{i,j} = \phi_j(x_i)
    \end{align*}
    where $x_i\in \domain{f}$ and $\phi_j$ are the characters of the non-zero Fourier coefficients of $p$.
\end{definition}

We emphasize that since the partial function $f$ does not have a canonical Fourier representation, we are using an approximating polynomial. Alternatively, one could obtain the domain evaluation matrix by fixing a representation of $f$. In the domain evaluation matrix $A_{D}(p)$, row $i$ is the vector of basis evaluations at input $x_i$ and column $j$ are the values of basis functions $\phi_{j}$ over all inputs in the domain. Therefore, $A_{D}(p)$ is a linear map between coefficient vectors and their corresponding polynomial evaluations. For example, given the Fourier coefficient vector of $p(x)$, defined as $a=(a_{1},\dots,a_{s})^{T}$, we have that $A_{D}(p)\, a=(p(x_1),\dots,p(x_m))^{T}$. Therefore, $A_{D}(p)$ encodes how the coefficients determine the polynomial values on the domain. Explicitly,

\begin{align*}A_{D}(p)=
    \begin{bmatrix}
        \phi_{1}(x_1) & \phi_{2}(x_{1}) & \cdots & \phi_{s}(x_{1})\\
        \phi_{1}(x_{2}) & \phi_{2}(x_{2}) & \cdots & \phi_{s}(x_{2})\\
        \cdots & \cdots & \cdots &\cdots\\
        \phi_{1}(x_{m}) & \phi_{2}(x_{m}) & \cdots & \phi_{s}(x_{m})
    \end{bmatrix}.
\end{align*}

Therefore, given the coefficient vector $a=(a_{1},\dots,a_{s})^{T}$ for $p(x)$,

\begin{align*}
    A_{D}(p)\,a=\begin{bmatrix}
        a_{1}\phi_{1}(x_{1})+a_{2}\phi_{2}(x_{1}) + \cdots + a_{s}\phi_{s}(x_{1})\\
        a_{1}\phi_{1}(x_{2}) + a_{2}\phi_{2}(x_{2}) + \cdots + a_{s}\phi_{s}(x_{2})\\
        \cdots\\
        a_{1}\phi_{1}(x_{m}) + a_{2}\phi_{2}(x_{m}) + \cdots + a_{s}\phi_{s}(x_{m})
    \end{bmatrix}
    =[p(x_{1}),\dots,p(x_{m})]^{T}.
\end{align*}

A natural object of study in this model are perturbations $\Delta\in \mathbb{R}^{s}$ to the original polynomial $p(x)=\sum_{j=1}^{s}a_{j}\phi_{j}(x)$ to get some perturbed polynomial,

\begin{align*}
    p_{\Delta}(x) \coloneqq p(x)+\sum_{j=1}^{s}\Delta_{j}\phi_{j}(x)=p(x)+L_{x}(\Delta),
\end{align*}

where for each $x, L_{x}(\Delta)=\sum_{j\in [s]}\Delta_{j}\phi_{j}(x)$. Let $B_x$ denote the set of \say{bad} perturbations,
\begin{align*}
    B_x &= \{\Delta \in \mathbb{R}^s: \abs{p_{\Delta}(x)} \leq \epsilon \}\\
    &= \{\Delta \in \mathbb{R}^s: L_x(\Delta)\in [-p(x) - \epsilon, -p(x) + \epsilon] \} .
\end{align*}

Ideally, we want to find a perturbation $\Delta$ such that $\Delta \not \in \bigcup_{x\in \{-1,1\}^n} B_x$ and $\forall x\in \domain{f}$, $p(x) = p_\Delta(x)$. This would ensure that we perturb the approximating polynomial without altering the output of $f$. To this end, we discuss the following helper lemmas. We use $b$ to represent $-p(x)$ so that the statements have full generality.

Let $V$ be a real vector space of finite dimension $d\geq 1$ and fix any inner product on $V$. For any nonzero linear functional $l: V\rightarrow \mathbb{R}$ and constant $b\in \mathbb{R}$ we let
\[
H(l,b)=\{v\in V: l(v)=b\}
\]
be the affine hyperplane defined by $l$ and $b$ and for $\epsilon>0$ we let
\[
S_{\epsilon}(l,b)=\{v\in V: |l(v)-b|\leq \epsilon\}
\]
be the corresponding (closed) $\epsilon$-slab. In all statements that follow, when we say measure we mean Lebesgue measure.

\begin{lemma}
\label{lem:hyper}
    Let $V$ be a $d$-dimensional real vector space with $d\geq 1$ and let $l_{1},\dots,l_{t}: V\rightarrow \mathbb{R}$ be nonzero linear functionals and let $b_{1},\dots,b_{t}\in \mathbb{R}$. Then,
    \begin{align*}
        V \setminus \bigcup_{i=1}^{t}H(l_{i},b_{i}) \neq \emptyset.
    \end{align*}
\end{lemma}

\begin{proof}
It is a standard fact that each such hyperplane $H(l_i, b_i)$ has measure zero in $\mathbb{R}^{d}$. The finite union of measure zero sets is also measure zero, and thus the finite union cannot equal $V$.
\end{proof}

We can generalize this lemma to all $\epsilon$-slabs.

\begin{lemma}
\label{lem:slabs}
    Let $V$ be a $d$-dimensional vector space with $d\geq 1$. Let $l_{1},\dots,l_{t}:V\rightarrow \mathbb{R}$ be nonzero linear functionals and let $b_{1},\dots,b_{t}\in\mathbb{R}$. Fix $\epsilon>0$. The finite union of $\epsilon$-slabs as previously defined cannot cover $V$. That is,
    \[
    V \setminus\bigcup_{i=1}^{t}S_{\epsilon}(l_{i},b_{i})\neq \emptyset.
    \]
\end{lemma}

\begin{proof}
    For each $i$ with $l_{i}\not\equiv0, \mathrm{Ker}(l_{i})$ is a hyperplane and therefore has measure $0$ in $\mathbb{R}^{d}$. The finite union $\bigcup_{i=1}^{t}\mathrm{Ker}(l_{i})$ has measure $0$ and so
    \[
    V \setminus \bigcup_{i=1}^{t}\mathrm{Ker}(l_{i})
    \]
    is non-empty. Thus there exists a $u\in V$ with $l_{i}(u)\not= 0$ for every $i$. For each $i$ consider the affine function $\phi_{i}(t)=l_{i}(tu)-b_{i}=tl_{i}(u)-b_{i}$ for the scalar $t$. Then since $l_{i}(u)\not=0$, we have that $|\phi_{i}(t)|\rightarrow \infty$ as $|t|\rightarrow \infty$. This for each $i$ there exists a $T_{i}>0$ such that for all $|t|>T_{i}$ we have $|\phi_{i}(t)|>\epsilon$, and so $tu\notin S_{\epsilon}(l_{i},b_{i})$. If we let $T=\max_{i}T_{i}$, then for all $|t|>T$, vector $tu$ avoids every slab $S_{\epsilon}(l_{i},b_{i})$. Therefore the union of slabs does not cover $V$.
\end{proof}

\Cref{lem:slabs} says that there exists some perturbation $\Delta\in\Ker(A_D(p))$ which can avoid all $\epsilon$-slabs, provided no $B_x$ contains $\Ker(A_D(p))$. Thus, it gives the perturbed polynomial $p_{\Delta}(x)$ the necessary margin away from $0$ in order to use $p_{\Delta}(x)$ to define a completion.
It is important to emphasize that the vector in \Cref{lem:slabs} exists because we are able to arbitrarily stretch it out enough in order to avoid all $\epsilon$ slabs.
However, if we want to use $\Delta$ to define $p_{\Delta}(x)$, we need to ensure that its norm remains polynomially-related to $\approxdegree(f)$. A visualization of this idea is illustrated in \Cref{fig:pert}. To this end, we give the following lemma.

\begin{figure}[ht]
    \centering
    \input{perturbation_vis}
    \caption{The ball here has radius $R\leq \poly(d)$, and we see two example $\epsilon$-slabs which would correspond to two points $x\notin \domain{f}$. The vector $\Delta$ shown here avoids these $\epsilon$-slabs while simultaneously being bounded in magnitude within the volume of the ball. This is what we will then choose as our perturbation for the polynomial $p_{\Delta}(x)$.}
    \label{fig:pert}
\end{figure}
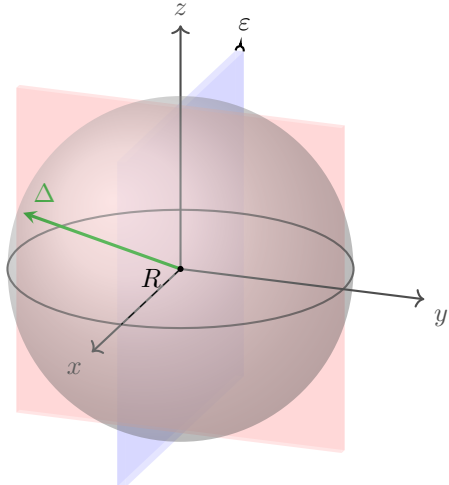

\begin{lemma}\label{lem:chebyshev}
    Consider $f:\domain{f}\rightarrow \{-1,1\}$ and let $p$ be an approximating polynomial of $f$. Suppose there exists some $\Delta \in \Ker(A_D(p))$ such that $\norm{\Delta}_{\infty} \polyleq \deg(p)$ and $\Delta \not\in \bigcup_{x\in \{-1,1\}^n} B_x$. Then $\D(f) \polyeq \Q(f)$.
\end{lemma}
\begin{proof}
    Let $M=\norm{\Delta}_{\infty}$ and $\gamma$ be the largest constant such that for all $x\in\{-1,1\}^n$, $\abs{p_{\Delta}(x)} \geq \gamma$. As $\Delta \not\in \bigcup_{x\in \{-1,1\}^n} B_x$, we may assume that $\gamma \polyeq \deg(p)^{-1}$. Furthermore, let $P: \mathbb{R}\rightarrow \mathbb{R}$ be the polynomial in \Cref{lem:chebyshev_sign} where $\epsilon = \frac{1}{3}$, $\delta = \frac{\gamma}{M}$ and $\deg(P) = \poly(\deg(p))$. Define the approximating polynomial $\tilde{p}(x)$ as,
    \begin{align*}
        \tilde{p}(x) \coloneqq P\left(\frac{p_{\Delta}(x)}{\norm{\Delta}_{\infty}}\right).
    \end{align*}
    Observe that for all $x\in\{-1,1\}^n$, since $\abs{p_{\Delta}(x)}\geq \gamma$, we have that $\abs{\frac{p_{\Delta}(x)}{M}} \in [\delta, 2]$. Therefore, by \Cref{lem:chebyshev_sign} $\abs{\tilde{p}(x)} \in [\frac{2}{3}, 1]$. Furthermore, since $\Delta\in \Ker(A_D(p))$, we find that $\textup{sign}(p(x)) = \textup{sign}(\tilde{p}(x))$. Letting $F(x) = \textup{sign}(\tilde{p}(x))$ be a natural completion of $f$, as $\abs{p_{\Delta}(x)}\geq \gamma$, we have that $\tilde{p}(x)$ is an approximating polynomial of $F$ within error $\frac{1}{3}$. Furthermore, $\deg(\tilde{p}) = \deg(P)\cdot \deg(p) \polyeq \deg(p)$. Therefore, by \Cref{lem:degree_lower_bound} and \Cref{lemma:general}, $\D(f)\polyeq \Q(f)$.
\end{proof}

In \Cref{lem:chebyshev}, we assumed that $\Delta \in \Ker(A_D(p))$. This way we ensure that for $x\in\domain{f}$, $p(x)=p_{\Delta}(x)$. However, in order to ensure that $\tilde{p}(x)$ is a valid approximating polynomial for some natural completion $F$, we may replace the requirement that $\Delta \in \Ker(A_D(p))$ with the relaxed version that $\textup{sign}(p(x)) = \textup{sign}(p_{\Delta}(x))$. It is also useful to note that in our discussion thus far, we have only considered the case where we perturb the original Fourier coefficients for a particular approximating polynomial $p(x)$ for a partial function $f$. However, one can easily generalize this thinking to the case where we pick some cutoff level $d'=\poly(\approxdegree(f))$ and add back all $\sum_{i=0}^{d'}\binom{n}{i}$ basis vectors $\chi(x)$ of degree at most $d'$ with Fourier coefficients of $0$, and then continue looking for a perturbation in the way we have described above.  We conjecture that this relaxation on $\Delta$ exactly characterizes when the natural completion of $f$ is possible.

\begin{conjecture}
    A partial function $f: \domain{f} \rightarrow \{-1,1\}$ with an approximating polynomial $p(x)$ of degree $d$ admits the natural completion $F$ (for approximate degree) if and only if there exists a perturbation $\Delta$ such that $p_{\Delta}(x)$ is a degree $\poly(d)$ approximating polynomial for $F$.
\end{conjecture}

Observe that per our discussion above, the second condition implies the first. We note that none of the ideas discuss or techniques used in this section were particularly constrained to multilinear polynomials and the ideas and techniques presented in this section may be of independent interest to fields in pure mathematics such as approximation and perturbation theory. The matrix $A_{D}(p)$ was partially formed from the basis functions of the Fourier representation of Boolean functions, but any general linear map could be considered and the same techniques from all the above lemmas would apply. Notably, the techniques in this section are immediately applicable to univariate polynomials as well as any polynomials in $n$ indeterminates $\mathbb{R}[x_{1},\dots,x_{n}]$ and the univariate case in particular may be of independent interest in the field of approximation theory with respect to polynomial approximation as constraint satisfaction problems.

\subsubsection{Hardness of Perturbing}

We now prove some hardness results for a generalization of the previously discussed perturbation problem. We begin by defining the \emph{Perturbation Finding} (\pf) problem.

\begin{problem}[\pf]\label{prob:pf}
    Suppose we are given multilinear polynomials $p,q_{1},\dots,q_{n}: \{-1,1\}^{t}\rightarrow \mathbb{R}$ given by their full tables of binary-encoded rational values, and rational constants $\epsilon, B > 0$ encoded in binary. The $\pf(\epsilon, B)$ problem asks whether there exists some vector $\Delta \in \mathbb{R}^n$ such that for all $x\in \{-1,1\}^t$,
    \begin{align*}
        \abs*{p(x) + \sum_{j=1}^{n}\Delta_{j}q_{j}(x)}&\geq \epsilon,\\
        \norm{\Delta}_{\infty}&\leq B.
    \end{align*}
    We let $\pf = \pf(1,1)$.
\end{problem}

\begin{observation}
    For any positive rational $(\epsilon,B)$ encoded in binary, $\pf(1,1)$ is polynomial-time equivalent to $\pf(\epsilon, B)$.
\end{observation}
\begin{proof}
    Given an instance $A = (p,q_1, \dots q_n)$ of \pf, the following input $A^\prime = (p^\prime, q^\prime_1,\dots q^\prime_n)$ is an instance of $\pf(\epsilon, B)$,
    \begin{align*}
        p^\prime(x) &= \epsilon\cdot p(x),\\
        q_i^\prime(x) &= \frac{\epsilon}{B}\cdot q_i(x).
    \end{align*}
    Furthermore, we claim that $\Delta$ solves $A$ if and only if $\Delta^\prime = B \Delta$ solves $A^\prime$. First, $\norm{\Delta}_{\infty}\leq 1$ if and only if $\norm{\Delta^\prime}_{\infty}\leq B$. Secondly, notice that,
    \begin{align*}
        p^\prime(x) + \sum_{j=1}^{n}\Delta_{j}^\prime q^\prime_{j}(x) &= p(x)\epsilon + \sum_{j=1}^{n} B\Delta_{j} \cdot \frac{\epsilon}{B} q_{j}(x)\\
        &=\epsilon \left(p(x) + \sum_{j=1}^{n} \Delta_{j} \cdot q_{j}(x)\right).
    \end{align*}
    Therefore, the $\epsilon$-condition holds in $A$ if and only if it holds in $A^\prime$.
\end{proof}

\begin{claim}
\label{clm:hardness_ab}
    Let $m=\poly(n)$ and $N=n+m$. Let $A\in \mathbb{Q}^{N\times n}$ and $b\in\mathbb{Q}^N$ be inputs which may be described using $\poly(n)$ bits. We refer to the problem of deciding whether there exists a vector $\Delta\in \mathbb{R}^n$ such that for every $r\in [N]$,
    \begin{align*}
        \abs{(A \Delta + b)_r} &\geq 1 \\
        \norm{\Delta}_{\infty} &\leq 1
    \end{align*}
    as Linear Margin Avoidance $(\lma)$. \lma is \NP-hard.
\end{claim}

\begin{proof}
    We proceed via a reduction from \sat. Let $\Phi = C_1 \land \cdots \land C_m$ be a $3$-CNF formula over variables $u_1,\dots u_n$. We define $A$ and $b$ where $N = m+n$ as follows. For variable $i\in [n]$, we define the first $n$ rows of $A$ as follows: for row $i\in [n]$ of $A$ set it to be $e_i$ and set $b_i = 0$, meaning the condition of row $i$ becomes $\abs{\Delta_i} \geq 1$. Therefore, as $\norm{\Delta}_{\infty} \leq 1$, we enforce $\Delta_i \in \{-1,1\}$.

    Next, for $j \in [m]$, each $n+j$ row in $A$ and $b$ represents the conditions for each clause $C_j$. Specifically, over the variables $u_i$, the row $a_{n+j}$ is defined as,
    \begin{align*}
        (a_{n+j})_i = \begin{cases}
        1 &\text{if $u_i\in C_j$}\\
        -1 &\text{if $\neg u_i\in C_j$}\\
        0 &\text{otherwise.}
        \end{cases}
    \end{align*}
    Furthermore, for all $j \in [m]$, $b_{n+j} = 3$. Thus the clause constraint at row $n+j$ becomes,
    \begin{align}
        \abs{a_{n+j}\Delta + 3 } \geq 1.\label{eq:clause_constraint}
    \end{align}
    If $C_j$ is unsatisfied, we have that $a_{n+j}\Delta = -3$ and $a_{n+j}\Delta \geq -2$ otherwise, meaning that \Cref{eq:clause_constraint} is satisfied if and only if $C_j$ is satisfied. Thus, every satisfying assignment of $\Phi$ corresponds to a feasible $\Delta\in \mathbb{R}^n$ and vice-versa, completing the proof.
\end{proof}

\begin{theorem}
    \pf is \NP-complete.
\end{theorem}
\begin{proof}
    Given a fixed $\Delta$, let $p_{\Delta}(x) =p(x) + \sum_{j=1}^{n}\Delta_{j}q_{j}(x)$.
    First, notice that the input can be described in $\poly(n, 2^t)$ characters. Therefore, we may check a yes-instance by receiving a vector $\Delta$ and check it for all $2^t$ inputs $x$, each requiring $\poly(2^t)$ time. Therefore, $\pf\in\NP$. 
    
    We now reduce \lma to \pf, which implies that \pf is \NP-hard. Let $t = \lceil \log N\rceil$ and define the subset of inputs $Z = \{z_r\}_{r\in [N]} \subseteq \{-1,1\}^t$. We note that $Z$ depends on $N$ and if $N$ is a power of $2$, then we can simply use this to cover all of the points, but in the case that it is not, we may define these extra points. We define the polynomials $p, q_1,\dots q_n$ as follows,
    \begin{align*}
        q_j(x) &= \begin{cases}
            A_{r,j} &\text{if $x=z_r\in Z$}\\
            0 &x\notin Z,
        \end{cases}\\
        p(x) &= \begin{cases}
            b_r &\text{if $x=z_r\in Z$}\\
            2 &x\notin Z.
        \end{cases}
    \end{align*}
    Therefore, $p_{\Delta}(x)$ becomes,
    \begin{align*}
        p_{\Delta}(x)= \begin{cases}
            (A\Delta + b)_r &\text{if $x=z_r\in Z$}\\
            2 &x\notin Z.
        \end{cases}
    \end{align*}
    Therefore, we get that $\abs{p_{\Delta}(x)} \geq 1$ for all $x\in \{-1,1\}^t$ if and only if $\abs{(A\Delta + b)_r} \geq 1$ for all $r\in [N]$. As $\Delta$ has the same condition $\norm{\Delta}_{\infty} \leq 1$ in both instances, we have reduced the problem in \Cref{clm:hardness_ab} to \pf.
\end{proof}

The hardness of the \pf problem concerns the general perturbation decision problem defined above. One may therefore inquire about the hardness of any completion, whether it be natural, na\"ive, or any other completion one could put forth. We note here that if we loosen one of the restrictions for finding a valid completion, this problem can be connected to a well-known conjecture in complexity theory. Namely, thus far in our discussion, we have required that all of our completions satisfy the condition that if $F$ is a completion of some partial function $f$, then $F(x)=f(x)$ for all $x$ in the domain of $f$. However, we may relax this to the case where we only require to match the partial function $f$ on a subset of inputs in its domain. Specifically, one may consider the case where given some polynomial $p(x)$ where $|p(x)|\geq 1/\poly(\approxdegree(f))$, that $\textup{sign}(p_{\Delta}(x)) = f(x)$ on a $(1-\delta)$-fraction of $x\in\domain{f}$. This is closely related to the formulation of the Aaronson-Ambainis conjecture~\cite{Need_For_Structure}. Let us recall the conjecture for convenience.

\begin{conjecture}[Aaronson-Ambainis conjecture; Conjecture $7.1$ in~\cite{Need_For_Structure}]
    Let $S\subseteq \{-1,1\}^{N}$ with $|S|\geq c2^{N}$, and let $f:S \rightarrow \{-1,1\}$. Then there exists a deterministic classical algorithm that makes $\poly(\Q(f),1/\alpha,1/c)$ queries, and that computes $f(X)$ on at least a $1-\alpha$ fraction of $X\in S$.
\end{conjecture}

In our completion framework the corresponding approximate degree question is as follows.

\begin{conjecture}
    Let $f:\domain{f} \rightarrow \{-1,1\}$ be a partial function with approximate degree $d$ where $|\domain{f}|\geq\epsilon2^{n}$. Then there exists a completion $F$\footnote{Technically, the completion is from a subfunction of $f$ and not the original $f$ itself.} of $f$ where $F(x)=f(x)$ on $(1-\delta)|\domain{f}|$ many inputs $x\in \domain{f}$ and set arbitrarily otherwise, with $\approxdegree(F)=\poly(d,1/\epsilon,1/\delta)$.
\end{conjecture}

We note that our condition is slightly different in the sense that we consider the relationship between $\D(f)$ and $\approxdegree(f)$ as opposed to $\D(f)$ versus $\Q(f)$.

%% file: perturbation_vis.tex
\tdplotsetmaincoords{70}{110}

\begin{tikzpicture}[tdplot_main_coords, scale=1.15]

\def\sphereradius{2}
\def\epsilon{0.1}

\fill[blue!20, opacity=0.5] 
    ({-\sphereradius},{-\epsilon/2},{-\sphereradius}) -- 
    ({\sphereradius},{-\epsilon/2},{-\sphereradius}) -- 
    ({\sphereradius},{-\epsilon/2},{\sphereradius}) -- 
    ({-\sphereradius},{-\epsilon/2},{\sphereradius}) -- cycle;

\fill[blue!20, opacity=0.5] 
    ({-\sphereradius},{\epsilon/2},{-\sphereradius}) -- 
    ({\sphereradius},{\epsilon/2},{-\sphereradius}) -- 
    ({\sphereradius},{\epsilon/2},{\sphereradius}) -- 
    ({-\sphereradius},{\epsilon/2},{\sphereradius}) -- cycle;

\fill[blue!20, opacity=0.5] 
    ({\sphereradius},{-\epsilon/2},{-\sphereradius}) -- 
    ({\sphereradius},{\epsilon/2},{-\sphereradius}) -- 
    ({\sphereradius},{\epsilon/2},{\sphereradius}) -- 
    ({\sphereradius},{-\epsilon/2},{\sphereradius}) -- cycle;

\fill[blue!20, opacity=0.5] 
    ({-\sphereradius},{-\epsilon/2},{-\sphereradius}) -- 
    ({-\sphereradius},{\epsilon/2},{-\sphereradius}) -- 
    ({-\sphereradius},{\epsilon/2},{\sphereradius}) -- 
    ({-\sphereradius},{-\epsilon/2},{\sphereradius}) -- cycle;

\fill[red!20, opacity=0.5] 
    ({-\epsilon/2},{-\sphereradius},{-\sphereradius}) -- 
    ({-\epsilon/2},{\sphereradius},{-\sphereradius}) -- 
    ({-\epsilon/2},{\sphereradius},{\sphereradius}) -- 
    ({-\epsilon/2},{-\sphereradius},{\sphereradius}) -- cycle;

\fill[red!20, opacity=0.5] 
    ({\epsilon/2},{-\sphereradius},{-\sphereradius}) -- 
    ({\epsilon/2},{\sphereradius},{-\sphereradius}) -- 
    ({\epsilon/2},{\sphereradius},{\sphereradius}) -- 
    ({\epsilon/2},{-\sphereradius},{\sphereradius}) -- cycle;

\fill[red!20, opacity=0.5] 
    ({-\epsilon/2},{\sphereradius},{-\sphereradius}) -- 
    ({\epsilon/2},{\sphereradius},{-\sphereradius}) -- 
    ({\epsilon/2},{\sphereradius},{\sphereradius}) -- 
    ({-\epsilon/2},{\sphereradius},{\sphereradius}) -- cycle;

\fill[red!20, opacity=0.5] 
    ({-\epsilon/2},{-\sphereradius},{-\sphereradius}) -- 
    ({\epsilon/2},{-\sphereradius},{-\sphereradius}) -- 
    ({\epsilon/2},{-\sphereradius},{\sphereradius}) -- 
    ({-\epsilon/2},{-\sphereradius},{\sphereradius}) -- cycle;


\draw[black!80, thick, opacity=0.7] (0,0,0) circle (\sphereradius);

\draw[thick,->, black!70] (0,0,0) -- (3,0,0) node[anchor=north east]{$x$};
\draw[thick,->, black!70] (0,0,0) -- (0,3,0) node[anchor=north west]{$y$};
\draw[thick,->, black!70] (0,0,0) -- (0,0,3) node[anchor=south]{$z$};

\draw[-stealth, very thick, green!60!black] (0,0,0) -- (1.42,-1.42,1) node[anchor=south west]{$\Delta$};

\draw[thick, decoration={brace, amplitude=3pt}, decorate] 
    ({-\sphereradius},{-\epsilon/2},{\sphereradius}) -- ({-\sphereradius},{\epsilon/2},{\sphereradius}) 
    node[midway, above=10pt, left=-8pt]{$\varepsilon$};

      \shade[ball color = gray!40, opacity = 0.4] (0,0) circle (2cm);
  \fill[fill=black] (0,0) circle (1pt);
  \draw[dashed] (0,0 ) -- node[above]{$R$} (2,0);

\end{tikzpicture}

%% file: appendix.tex
\section{Improvements on sculpting results}\label{section:sculpting_results}

\begin{definition}[Fractional Certificate Complexity] For any (partial) Boolean function $f$,
\[
\begin{aligned}
     \FC(f) = & \max_{x\in\Dom(f)} \min_w \sum_{i \in [n]} w(x,i)  \\
    &\textit{s.t.} \sum_{i : x_i \neq y_i} w(x,i) \geq 1 : \forall y \in \Dom(f) \text{ s.t. } f(x) \neq f(y).\\
    & \textit{and  } w(x,i) \geq 0 & \forall x \in \Dom(f), \forall i \in [n].
\end{aligned}
\]
\end{definition}

Here we present a tighter bound than Theorem 1.3 of~\cite{aaronson2015sculptingquantumspeedups} which replaces the $\Q(f)$ by $\approxdegree(f)$.

\begin{lemma}\label{lemma:appendix}
    For every (possible partial) Boolean function $f$, $\D(f) = O(\approxdegree (f)^2 \cdot \log^2 |\dom(f)| )$.
\end{lemma}

\begin{proof}
From~\cite{aaronson2015sculptingquantumspeedups}, $\D(f) = O(\D(f_{a,b}) \cdot \text{Bal}(f))$. Here $f_{a,b}$ is the function that distinguishes an input $a$ from domain $b$. This argument is similar to Theorem 10~\cite{aaronson2015sculptingquantumspeedups}. 

Now $\D(f_{a,b}) = \certificate{f_{a,b}}$. (from the observation of the proof of Theorem 3 on page 21 in~\cite{aaronson2015sculptingquantumspeedups} that certificate complexity of distinguishing an input matches the deterministic query complexity).

We need to show a bound on $\certificate{f_{a,b}}$. Notice that $\certificate{f_{a,b}} \leq \RC(f,x) \cdot \log |\dom(f)|$ follows from the Lemma 21 of~\cite{aaronson2015sculptingquantumspeedups}.
Next we show $\RC(f) = \fbs(f) = \FC(f)$. This follows from~\cite{kulkarni2016fractional,aaronson2008quantum} as $\fbs(f,x) = \FC(f,x)$ for all $x$ and $\RC(f) = \FC(f)$. Thus, $\certificate{f_{a,b}} \leq \fbs(f) \cdot \log | \dom(f)|$ since each of these measures is a maximization over all inputs $x$.
Finally, $\fbs(f) \leq \approxdegree(f)^2$ by Lemma 28 of~\cite{anshu2020querytocommunicationliftingadversarybounds} and for every $f$, $Bal(f) \leq \log |\dom(f)|$. Thus we obtain the desired result.

\end{proof}

\begin{corollary}
\label{cor:mdegdomain}
    For any such $f$ and any measure $M$ polynomially related to $\D$ on total functions, $\cmplcom{M}(f) \polyleq \approxdegree (f) \cdot \log |\dom(f)|$.
\end{corollary}

\begin{proof}
 Combining \Cref{lem:polymeasures,fact:dequalscompletion} with \Cref{lemma:appendix} we get the result.
\end{proof}

In particular, we get $\cmplcom{\approxdegree}(f) \polyeq \approxdegree (f) \cdot \log |\dom(f)|$.

%% file: main.bbl
\newcommand{\etalchar}[1]{$^{#1}$}
\begin{thebibliography}{ABDK{\etalchar{+}}21a}

\bibitem[AA14]{Need_For_Structure}
Scott Aaronson and Andris Ambainis.
\newblock The need for structure in quantum speedups.
\newblock {\em Theory of Computing}, 10(6):133--166, 2014.

\bibitem[Aar08]{aaronson2008quantum}
Scott Aaronson.
\newblock Quantum certificate complexity.
\newblock {\em Journal of Computer and System Sciences}, 74(3):313--322, 2008.

\bibitem[Aar21]{aaronson2021open}
Scott Aaronson.
\newblock Open problems related to quantum query complexity.
\newblock {\em ACM Transactions on Quantum Computing}, 2(4):1--9, 2021.

\bibitem[ABD16]{aaronson2015sculptingquantumspeedups}
Scott Aaronson and Shalev Ben-David.
\newblock Sculpting quantum speedups.
\newblock In {\em 31st Conference on Computational Complexity (CCC 2016)},
  pages 26--1. Schloss Dagstuhl--Leibniz-Zentrum f{\"u}r Informatik, 2016.

\bibitem[ABDK{\etalchar{+}}21a]{aaronson2021degree}
Scott Aaronson, Shalev Ben-David, Robin Kothari, Shravas Rao, and Avishay Tal.
\newblock Degree vs. approximate degree and quantum implications of huang’s
  sensitivity theorem.
\newblock In {\em Proceedings of the 53rd Annual ACM SIGACT Symposium on Theory
  of Computing}, pages 1330--1342, 2021.

\bibitem[ABDK21b]{anshu2020querytocommunicationliftingadversarybounds}
Anurag Anshu, Shalev Ben-David, and Srijita Kundu.
\newblock On query-to-communication lifting for adversary bounds.
\newblock In {\em 36th Computational Complexity Conference (CCC 2021)}, pages
  30--1. Schloss Dagstuhl--Leibniz-Zentrum f{\"u}r Informatik, 2021.

\bibitem[Amb18]{ambainis2018understanding}
Andris Ambainis.
\newblock Understanding quantum algorithms via query complexity.
\newblock In {\em Proceedings of the International Congress of Mathematicians:
  Rio de Janeiro 2018}, pages 3265--3285. World Scientific, 2018.

\bibitem[BBC{\etalchar{+}}01]{beals2001quantum}
Robert Beals, Harry Buhrman, Richard Cleve, Michele Mosca, and Ronald De~Wolf.
\newblock Quantum lower bounds by polynomials.
\newblock {\em Journal of the ACM (JACM)}, 48(4):778--797, 2001.

\bibitem[Bd02]{BUHRMAN200221}
Harry Buhrman and Ronald {de Wolf}.
\newblock Complexity measures and decision tree complexity: a survey.
\newblock {\em Theoretical Computer Science}, 288(1):21--43, 2002.
\newblock Complexity and Logic.

\bibitem[BD16]{bendavid2014structurepromisesquantumspeedups}
Shalev Ben-David.
\newblock The structure of promises in quantum speedups.
\newblock In {\em 11th Conference on the Theory of Quantum Computation,
  Communication and Cryptography}, 2016.

\bibitem[BD20]{Shalev}
Shalev Ben-David.
\newblock Polynomials, part 1: Symmetrization, 2020.

\bibitem[BDCG{\etalchar{+}}24]{ben2024symmetries}
Shalev Ben-David, Andrew~M Childs, Andr{\'a}s Gily{\'e}n, William Kretschmer,
  Supartha Podder, and Daochen Wang.
\newblock Symmetries, graph properties, and quantum speedups.
\newblock {\em SIAM Journal on Computing}, 53(6):FOCS20--368, 2024.

\bibitem[BDK24]{ben2024separations}
Shalev Ben-David and Srijita Kundu.
\newblock Separations in query complexity for total search problems.
\newblock {\em arXiv preprint arXiv:2410.16245}, 2024.

\bibitem[BDKP16]{balaji2015graph}
Nikhil Balaji, Samir Datta, Raghav Kulkarni, and Supartha Podder.
\newblock Graph properties in node-query setting: Effect of breaking symmetry.
\newblock In Piotr Faliszewski, Anca Muscholl, and Rolf Niedermeier, editors,
  {\em 41st International Symposium on Mathematical Foundations of Computer
  Science, {MFCS} 2016, Krak{\'{o}}w, Poland, August 22-26, 2016}, volume~58 of
  {\em LIPIcs}, pages 17:1--17:14. Schloss Dagstuhl - Leibniz-Zentrum f{\"{u}}r
  Informatik, 2016.

\bibitem[BDW02]{buhrman2002complexity}
Harry Buhrman and Ronald De~Wolf.
\newblock Complexity measures and decision tree complexity: a survey.
\newblock {\em Theoretical Computer Science}, 288(1):21--43, 2002.

\bibitem[BT22]{bun2022approximate}
Mark Bun and Justin Thaler.
\newblock Approximate degree in classical and quantum computing.
\newblock {\em Coding for Interactive Communication: A Survey},
  15(3-4):229--423, 2022.

\bibitem[CGL{\etalchar{+}}23]{chakraborty2022certificate}
Sourav Chakraborty, Anna G\'{a}l, Sophie Laplante, Rajat Mittal, and Anupa
  Sunny.
\newblock {Certificate Games}.
\newblock In Yael Tauman~Kalai, editor, {\em 14th Innovations in Theoretical
  Computer Science Conference (ITCS 2023)}, volume 251 of {\em Leibniz
  International Proceedings in Informatics (LIPIcs)}, pages 32:1--32:24,
  Dagstuhl, Germany, 2023. Schloss Dagstuhl -- Leibniz-Zentrum f{\"u}r
  Informatik.

\bibitem[Cha19]{chailloux2019note}
Andr{\'e} Chailloux.
\newblock A note on the quantum query complexity of permutation symmetric
  functions.
\newblock In {\em ITCS 2019-10th Annual Innovations in Theoretical Computer
  Science}, 2019.

\bibitem[DFH19]{biased_agreement_test}
Irit Dinur, Yuval Filmus, and Prahladh Harsha.
\newblock Analyzing boolean functions on the biased hypercube via
  higher-dimensional agreement tests.
\newblock In {\em Proceedings of the Thirtieth Annual ACM-SIAM Symposium on
  Discrete Algorithms}, SODA '19, page 2124–2133, USA, 2019. Society for
  Industrial and Applied Mathematics.

\bibitem[DFH24]{Dinur_2024}
Irit Dinur, Yuval Filmus, and Prahladh Harsha.
\newblock Sparse juntas on the biased hypercube.
\newblock {\em TheoretiCS}, Volume 3, July 2024.

\bibitem[DFL{\etalchar{+}}21]{DFLLV21}
Neta Dafni, Yuval Filmus, Noam Lifshitz, Nathan Lindzey, and Marc Vinyals.
\newblock {Complexity Measures on the Symmetric Group and Beyond}.
\newblock In James~R. Lee, editor, {\em 12th Innovations in Theoretical
  Computer Science Conference (ITCS 2021)}, volume 185 of {\em Leibniz
  International Proceedings in Informatics (LIPIcs)}, pages 87:1--87:5,
  Dagstuhl, Germany, 2021. Schloss Dagstuhl -- Leibniz-Zentrum f{\"u}r
  Informatik.

\bibitem[DJ92]{deutsch_jozsa}
David Deutsch and Richard Jozsa.
\newblock Rapid solution of problems by quantum computation.
\newblock {\em Proceedings of the Royal Society of London. Series A:
  Mathematical and Physical Sciences}, 439(1907):553--558, 12 1992.

\bibitem[Gav25]{gavinsky2025unambiguous}
Dmytro Gavinsky.
\newblock Unambiguous parity-query complexity.
\newblock {\em Random Structures \& Algorithms}, 66(3):e70010, 2025.

\bibitem[GHR{\etalchar{+}}25]{goos2025pseudodeterministic}
Mika G{\"o}{\"o}s, Nathaniel Harms, Artur Riazanov, Anastasia Sofronova, Dmitry
  Sokolov, and Weiqiang Yuan.
\newblock Pseudodeterministic communication complexity.
\newblock {\em arXiv preprint arXiv:2511.04794}, 2025.

\bibitem[GSLW19]{Gily_n_2019}
András Gilyén, Yuan Su, Guang~Hao Low, and Nathan Wiebe.
\newblock Quantum singular value transformation and beyond: exponential
  improvements for quantum matrix arithmetics.
\newblock In {\em Proceedings of the 51st Annual ACM SIGACT Symposium on Theory
  of Computing}, STOC ’19, page 193–204. ACM, June 2019.

\bibitem[HN12]{huynh2012virtue}
Trinh Huynh and Jakob Nordstrom.
\newblock On the virtue of succinct proofs: Amplifying communication complexity
  hardness to time-space trade-offs in proof complexity.
\newblock In {\em Proceedings of the forty-fourth annual ACM symposium on
  Theory of computing}, pages 233--248, 2012.

\bibitem[Hua19]{Hua19}
Hao Huang.
\newblock Induced subgraphs of hypercubes and a proof of the sensitivity
  conjecture.
\newblock {\em Ann. Math.}, 190(3):949, November 2019.

\bibitem[KLLM23]{keevash2021globalhypercontractivityapplications}
Peter Keevash, Noam Lifshitz, Eoin Long, and Dor Minzer.
\newblock Hypercontractivity for global functions and sharp thresholds.
\newblock {\em J. Amer. Math. Soc.}, July 2023.

\bibitem[KM25]{khot2025biasedlinearitytesting1}
Subhash Khot and Kunal Mittal.
\newblock {Biased Linearity Testing in the 1\% Regime}.
\newblock In Srikanth Srinivasan, editor, {\em 40th Computational Complexity
  Conference (CCC 2025)}, volume 339 of {\em Leibniz International Proceedings
  in Informatics (LIPIcs)}, pages 10:1--10:23. Schloss Dagstuhl --
  Leibniz-Zentrum f{\"u}r Informatik, 2025.

\bibitem[KP16]{kulkarni2016quantum}
Raghav Kulkarni and Supartha Podder.
\newblock Quantum query complexity of subgraph isomorphism and homomorphism.
\newblock In {\em 33rd Symposium on Theoretical Aspects of Computer Science},
  2016.

\bibitem[KT16]{kulkarni2016fractional}
Raghav Kulkarni and Avishay Tal.
\newblock On fractional block sensitivity.
\newblock {\em Chicago J. Theor. Comput. Sci}, 8:1--16, 2016.

\bibitem[LM19]{noisey_p-biased}
Noam Lifshitz and Dor Minzer.
\newblock Noise sensitivity on the p -biased hypercube.
\newblock In {\em 2019 IEEE 60th Annual Symposium on Foundations of Computer
  Science (FOCS)}, pages 1205--1226, 2019.

\bibitem[LMZ21]{Larson_2021}
Jeffrey Larson, Matt Menickelly, and Baoyu Zhou.
\newblock Manifold sampling for optimizing nonsmooth nonconvex compositions.
\newblock {\em SIAM Journal on Optimization}, 31(4):2638–2664, January 2021.

\bibitem[Mag09]{magnani2009contactequationslipschitzextensions}
Valentino Magnani.
\newblock Contact equations, lipschitz extensions and isoperimetric
  inequalities, 2009.

\bibitem[MOO10]{mossel2005noisestabilityfunctionslow}
Elchanan Mossel, Ryan O'Donnell, and Krzysztof Oleszkiewicz.
\newblock Noise stability of functions with low influences: Invariance and
  optimality.
\newblock {\em Ann. Math.}, 171(1):295--341, March 2010.

\bibitem[Nis89]{nisan1989crew}
Noam Nisan.
\newblock Crew prams and decision trees.
\newblock In {\em Proceedings of the twenty-first annual ACM symposium on
  Theory of computing}, pages 327--335, 1989.

\bibitem[NR25]{nagda2025optimaldistinguishersplantedclique}
Ansh Nagda and Prasad Raghavendra.
\newblock On optimal distinguishers for planted clique, 2025.

\bibitem[O'D14]{odonnell2021analysisbooleanfunctions}
Ryan O'Donnell.
\newblock {\em Analysis of Boolean Functions}.
\newblock Cambridge University Press, 2014.

\bibitem[Pat92]{paturi1992degree}
Ramamohan Paturi.
\newblock On the degree of polynomials that approximate symmetric boolean
  functions (preliminary version).
\newblock In {\em Proceedings of the twenty-fourth annual ACM symposium on
  Theory of computing}, pages 468--474, 1992.

\bibitem[PYY25]{podder2025fine}
Supartha Podder, Penghui Yao, and Zekun Ye.
\newblock On the fine-grained query complexity of symmetric functions.
\newblock {\em computational complexity}, 34(1):1--77, 2025.

\bibitem[Sho99]{shor1999polynomial}
Peter~W Shor.
\newblock Polynomial-time algorithms for prime factorization and discrete
  logarithms on a quantum computer.
\newblock {\em SIAM review}, 41(2):303--332, 1999.

\bibitem[Sim97]{simon1997power}
Daniel~R Simon.
\newblock On the power of quantum computation.
\newblock {\em SIAM journal on computing}, 26(5):1474--1483, 1997.

\bibitem[SSSSS09]{ShalevShwartz2009StochasticCO}
Shai Shalev-Shwartz, Ohad Shamir, Nathan Srebro, and Karthik Sridharan.
\newblock Stochastic convex optimization.
\newblock In {\em Annual Conference Computational Learning Theory}, 2009.

\bibitem[ST20]{Sousi_2020}
Perla Sousi and Sam Thomas.
\newblock Cutoff for random walk on dynamical {E}rdős–{R}ényi graph.
\newblock {\em Annales de l’Institut Henri Poincaré, Probabilités et
  Statistiques}, 56(4), November 2020.

\bibitem[YZ24]{yamakawa2024verifiable}
Takashi Yamakawa and Mark Zhandry.
\newblock Verifiable quantum advantage without structure.
\newblock {\em Journal of the ACM}, 71(3):1--50, 2024.

\end{thebibliography}
